\documentclass[12pt]{article}
\usepackage{amsmath}
\usepackage{amsfonts,amsthm}
\usepackage{latexsym}
\usepackage{amscd}
\usepackage{amssymb}
\usepackage{graphicx}
\usepackage{float}
\usepackage[applemac]{inputenc}
\usepackage{dsfont}
\usepackage{multirow}
\usepackage{dcolumn}
\def \Beta {\boldsymbol{\beta}}
\usepackage{rccol}
\usepackage{subfigure}
\usepackage[pagewise]{lineno}
\usepackage{enumitem}
\usepackage{scalefnt}
\usepackage{hhline}
\usepackage{array}
\usepackage{supertabular}
\usepackage{booktabs}
\usepackage[colorlinks=true]{hyperref}
\usepackage{cleveref}
\crefname{subsection}{subsection}{subsections}
\numberwithin{equation}{section}
\usepackage{textcomp}
\usepackage{array}
\usepackage{supertabular}
\usepackage{hhline}
\DeclareSymbolFont{symbolsC}{U}{txsyc}{m}{n}
\DeclareMathSymbol{\ngg}{\mathrel}{symbolsC}{52}
\makeatletter
\newcommand\arraybslash{\let\\\@arraycr}
\makeatother
\newcolumntype{+}{>{\global\let\currentrowstyle\relax}}
\newcolumntype{^}{>{\currentrowstyle}}

\newlength{\bracewidth}

\def\fudge{\mathchoice{}{}{\mkern.5mu}{\mkern.8mu}}
\def\bbc#1#2{{\rm \mkern#2mu\vbar\mkern-#2mu#1}}
\def\bbb#1{{\rm I\mkern-3.5mu #1}}
\def\bba#1#2{{\rm #1\mkern-#2mu\fudge #1}}
\def\bb#1{{\count4=`#1 \advance\count4by-64 \ifcase\count4\or\bba A{11.5}\or
   \bbb B\or\bbc C{5}\or\bbb D\or\bbb E\or\bbb F \or\bbc G{5}\or\bbb H\or
   \bbb I\or\bbc J{3}\or\bbb K\or\bbb L \or\bbb M\or\bbb N\or\bbc O{5} \or
   \bbb P\or\bbc Q{5}\or\bbb R\or\bbc S{4.2}\or\bba T{10.5}\or\bbc U{5}\or
   \bba V{12}\or\bba W{16.5}\or\bba X{11}\or\bba Y{11.7}\or\bba Z{7.5}\fi}}
\def \R {\bb R}
\def \1 {\bb 1}

\def \n {\noindent}

\def \Beta {\boldsymbol{\beta}}
\restylefloat{float}

\newtheorem{theorem}{Theorem}[section]

\newtheorem{lemma}{Lemma}[section]

\title{Effects of migration on vector-borne diseases with forward and backward stage progression}

\author{Derdei Bichara\\
Department of Mathematics \& Center for Computational and
 Applied \\Mathematics, California State University, Fullerton, CA 92831, USA.}

\date{}

\begin{document}

\maketitle

\begin{abstract}\n 
Is it possible to break the host-vector chain of transmission when there is an influx of infectious hosts into a na\"{i}ve population and competent vector? To address this question, a class of vector-borne disease models with an arbitrary number of infectious stages that account for immigration of infective individuals is formulated. The proposed model accounts for forward and backward progression, capturing the mitigation and aggravation to and from any stages of the infection, respectively. The model has a rich dynamic, which depends on the patterns of infected immigrant influx into the host population and connectivity of the transfer between infectious classes. We provide conditions under which the answer of the initial question is positive. 
\end{abstract}



\section{Introduction}

Vector-borne diseases represent a major public health problem around the world, cause over one million deaths, one billion cases each year \cite{world2014global}, and more than half of the world's population is at risk \cite{world2004world}. They are typically associated with the tropics and subtropics where these diseases are endemic. However, recently these diseases have expanded their geographical distribution and have been reported in many temperate countries. For instance,  Dengue and Chikungunya have been reported in France \cite{gould2010first,grandadam2011chikungunya,paty2014large,vega2013high}, Italy \cite{poletti2011transmission, rezza2007infection}, and Portugal \cite{wilder20142012}. The CDC recently \cite{CDCVectorBorneUSA} reported that mosquito, tick, and flea bite borne diseases tripled in the United States from 2004 through 2016.\\

Many drivers are reported to be behind the geographic expansion of vector-borne diseases, including but not limited to trade, travel, climate change, urbanization and other social upheaval phenomena \cite{jones2008global,kilpatrick2012drivers,rogers2006climate}. 
Particularly, immigration and migration have been pointed to be the leading drivers in the emergence of vector-borne diseases in temperate nations \cite{barnett2008role}. For instance, a term \textit{Airport malaria} has been coined by \cite{isaacson1989airport}. Indeed, it describes a malaria infection that has resulted from the bite of an infected tropical anopheline mosquito  by persons whose geographic history excludes exposure to this vector in its natural habitat \cite{isaacson1989airport}.   Hereafter, we use the term ``immigrants" to represent both endemic area borne individuals migrating into a ``na\"{i}ve" area as well as non-endemic area native individuals who acquired an \textit{airport vector-borne disease} after a stint in an endemic area. Naturally, the latter term follows the definition of \textit{Airport malaria}.\\

 Moreover, the existence of vector populations that are capable of transmitting several arboviruses in the US and western Europe has been widely documented. For instance, the vector \textit{Aedes albopictus} also known as the Asian tiger mosquito, a vector competent of transmitting many arbovirus including Japanese encephalitis, West Nile, and yellow fever, Dengue, etc., to humans \cite{mitchell1995geographic,shroyer1986aedes}, is well established in North America \cite{benedict2007spread,mitchell1995geographic,shroyer1986aedes} and Europe \cite{mitchell1995geographic}.\\

Together, the increasing interconnectedness of the world brings an influx of viremic (latent or infectious) individuals and their epidemiological life-history into naive geographic areas or populations, and thereby infecting local competent vector populations. This could potentially create a chain reaction that could lead to an autochthonous transmission cycle of arboviruses and sporadic outbreaks of vector-borne diseases in these otherwise naive host populations. For instance, the presumed index case of Italy's 2007 Chikungunya outbreak was a man from India who developed symptoms while visiting relatives in one of the villages where the outbreak started \cite{rezza2007infection}. Similarly, an estimated 475 cases of imported chikungunya are reported in mainland France \cite{paty2014large} from November 2013 to June 2014, and these cases are reportedly traced back to travelers returning from the French Caribbean islands where chikungunya and Dengue are endemic \cite{paty2014large}. The 2012 Portugal's Dengue outbreak was reported to be imported by a traveler from Venezuela  \cite{wilder20142012}.  
It is therefore important the gauge the impacts of infected immigrants of the dynamics of vector-borne diseases. \\
 
Typically, modeling the dynamics of directly transmitted or vector-borne diseases have often been based on the assumption that the recruitment into the considered population is completely susceptible, and thereby sweeping the effects of global movement of individuals across the world at unprecedented levels under the rug. 
 To the best of our knowledge, Brauer and van den Driessche \cite{BraVdd01}  were the to first propose a mathematical model that accounts for  immigration that includes infected individuals using an SIS structure in an attempt to study HIV in prisons.  Subsequently, McCluskey and van den Driessche \cite{1056.92052} proposed a model studied with the same features where both immigration of latent and infective are considered, in the context tuberculosis. These two papers \cite{BraVdd01,1056.92052} showed that there is no disease-free equilibrium and that the endemic equilibrium (EE) is globally asymptotically stable. Indeed, the model proposed in  \cite{BraVdd01} is a two-dimensional, and the  Poincar\'e-Bendixson theorem has been used to study the global stability of the EE. In  \cite{1056.92052}, the authors  considered model is an $SEIS$ model and a geometric approach \cite{Limul95,MR95k:92022} is used to prove the global stability of the EE. 
Li et \textit{al.} \cite{guo2012impacts} generalized the before-mentioned models for staged-progression model -- by considering a model with $n$ infectious stages and a proportion $p_i$ of the total influx is incorporated for each infectious class $I_i$. The models in  the before-mentioned papers are all for directly transmitted diseases.\\

Recently, Tumwiine et \textit{al.} \cite{tumwiine2010host} investigated the effects of infected immigrants using an $SIR-SI$ Ross-Macdonal's model in the context of malaria and showed that the disease persists in host and vector populations whenever the proportion of infected immigrants is non zero. However, their model does not account of immigration of latent individuals, a critical category as these pass the precautionary measures of screenings, if these were in place. \\

In this paper, we derive and investigate the global behavior of a system that captures the dynamics of a class of vector-borne models accounting for flux of infected individuals at different stages of infection.  Of particular interest is the impact of the influx of infected individuals and transfer rates between infectious classes on the overall the dynamics of the model. The paper is organized as follows:\\

$\bullet$ We derive a class of vector-borne models with $n$ stages of infection, for which there is a flux of infected and infectious immigrants at all of these infectious stages. The formulated model accounts also for the progression and amelioration during the infectious stages, from an arbitrary stage $i$ to an arbitrary stage $j$.  The transfer is considered a progression if $i>j$ and an amelioration if $i<j$ (\Cref{sec:SPHostFlux}).\\

$\bullet$ We completely study the dynamics of the proposed model (\Cref{sec:GSVBAP}). It turns out that the model has a variety of dynamics, which depends on the patterns of the influx infected host into the population and the transfer rate matrix -- that describes the amelioration and deterioration of hosts' infectivity level. Particularly, we show that, under certain conditions, it is possible to corral the infectious hosts only into the classes in which they are replenished and maintain the vector populations disease-free. A threshold $\mathcal N^2(p,\mathbf{p},p_{n+1})$ plays a critical role for the existence of such steady-state. \\

$\bullet$ We provide the global dynamics of the model when there is no influx of infected individuals into the population, which surprisingly has not been done (\Cref{SharpThreshold}). In this case, the model exhibits the threshold phenomenon -- the basic reproduction number $\mathcal R_0^2$ determines the outcome of the disease both in host and vector population. It happens that $\mathcal R_0^2:=\mathcal N^2(0,\mathbf{0},0)$. \\

$\bullet$ Illustrations of the results and numerical simulations are carried out in \Cref{IllustrationsNSimulations}.

\section{Formulation of the model}\label{sec:SPHostFlux}
We consider a disease whose evolution is captured by a host-vector interaction for which the host population is composed of susceptible, exposed, recovered and infectious of stage $i$ ($1\leq i\leq n$). These subpopulations are denoted respectively by $S_h$, $E_h$, $R_h$ and $I_i$. The total host population is therefore $N_h=S_h+E_h+\sum_{i=1}^nI_i+R_h$. The vector population $N_v$ is composed of susceptible, exposed, and infectious arthropods; denoted by $S_v$, $E_v$, and $I_v$, respectively.\\

The total host population is replenished through a constant recruitment, $\pi_h$, that includes birth and migratory influx of individuals. Of this constant recruitment, a proportion $p$, $p_i$ ($i=1,\dots,n$) and $p_{n+1}$ is latent, infectious at stage $i$ and recovered, respectively. Thus, the total recruitment in the susceptible class is $\pi_h\left( 1-p-\sum_{i=1}^{n+1}p_i \right)$. Naturally, we assume that, $0\leq p\leq1$, $0\leq p_i\leq1$ and $0\leq p+\sum_{i=1}^{n+1}p_i<1$. \\
Susceptible hosts are infected at the rate $a\beta_{vh}\frac{I_v}{N_n}$, 
where $a$ is the biting/landing rate and $\beta_{vh}$ is the Host's infectiousness by arthropod per biting/landing. Vector's infection term is captured by 
 $a S_{v}\frac{\sum_{i=1}^n\beta_iI_i}{N_n},$ 
where $\beta_i$ is the vector's infectiousness by infected hosts of stage $i$. This accounts for the differential infectivity of vectors with respect to hosts' infectious stages.\\

Motivated by \cite{guo2012impacts,guo2012global,1008.92032}, we incorporate incremental and non-incremental amelioration and recrudescence in the infectiosity at each stage of the host's infection. For instance,  for models in \cite{BicharaIggidrSmith2017,cruz2012control,palmer2018dynamics}, the transitions between infection stages are incremental, that is, always from stage $i$ to $i+1$. However, with vector-borne diseases, a bite of infected arthropod to an already infected host, say at stage $i$, may increase this host's parasitemia, thereby catapulting its infectious class from stage $i$ to any stage, say $i+j$, where $1\leq i+j\leq n$. To incorporate this phenomenon, we denote by $\gamma_{ij}$, the \textit{per capita} rate at which the host progresses from stage $i$ to stage $j$. Similarly, the increase of treatments (which decrease the parasiteamia in the blood-stream) of vector-borne diseases could alleviate the host's infection and therefore, its stage could change form $i$ to $i-k$, where $ i-k\geq 0$. We denote by $\delta_{ij}$ the \textit{per capita} rate at which the host regresses from stage $i$ to stage $j$, where $j\leq i$.   These generalizations are illustrated in \Cref{fig:FlowSEIR1HostAmeliorationDeterioration}.\\

In concert, the overall dynamics of the Host-Vector infection is given by:
\begin{equation} \label{ModelAD}
\left\{\begin{array}{llll}
\displaystyle\dot S_{h}=\pi_h\left( 1-p-\sum_{i=1}^{n+1}p_i \right)-a\,\beta_{vh}\,S_{h}\,\dfrac{I_{v}}{N_{h}}-\mu_h S_h\\
\displaystyle\dot E_{h}=p\pi_h+ \,a\beta_{vh}\,S_{h}\,\dfrac{I_{v}}{N_{h}}-(\mu_h+\nu_h+\eta)E_h\\
\displaystyle\dot I_{1}=p_1\pi_h+\nu_hE_h-(\mu_h+\eta_1)I_1-I_{1}\sum_{j=2}^n\gamma_{1j}+\sum_{j=2}^n\delta_{j1}I_{j}\\
\displaystyle\dot I_{2}=p_2\pi_h-(\mu_h+\eta_2)I_2-I_{2}\left(\sum_{j=1,j<2}^n\delta_{2j}+\sum_{j=1,j>2}^n\gamma_{2j}\right)+\left(\sum_{j=1,j>2}^n\delta_{j2}I_j+\sum_{j=1,j<2}^n\gamma_{j2}I_j\right)\\ 
\vdots\\
\displaystyle\dot I_{i}=p_i\pi_h-(\mu_h+\eta_i)I_i-I_{i}\left(\sum_{j=1,j<i}^n\delta_{ij}+\sum_{j=1,j>i}^n\gamma_{ij}\right)+\left(\sum_{j=1,j>i}^n\delta_{ji}I_j+\sum_{j=1,j<i}^n\gamma_{ji}I_j\right)\\ 
\displaystyle\dot I_{n}=p_i\pi_h-(\mu_h+\eta_n)I_n-I_{n}\sum_{j=1}^{n-1}\gamma_{nj}+\sum_{j=1}^{n-1}\delta_{jn}I_{j}\\
\displaystyle \dot R_h=p_{n+1}\pi_h+\sum_{i=1}^n\eta_iI_{i}-\mu_hR\\
\displaystyle\dot S_{v}=\pi_v-a S_{v}\sum_{i=1}^n\dfrac{\beta_i I_{i}}{N_{h}}-(\mu_v+\delta_v)S_{v}\\
\displaystyle\dot E_{v}=a S_{v}\sum_{i=1}^n\dfrac{\beta_i I_{i}}{N_{h}}-(\mu_v+\nu_v+\delta_v)E_{v}\\
\displaystyle\dot I_{v}=\nu_vE_v-(\mu_v+\delta_v)I_{v}
\end{array}\right.
\end{equation}
\begin{figure}[ht]
\centering
\includegraphics[scale =0.85]{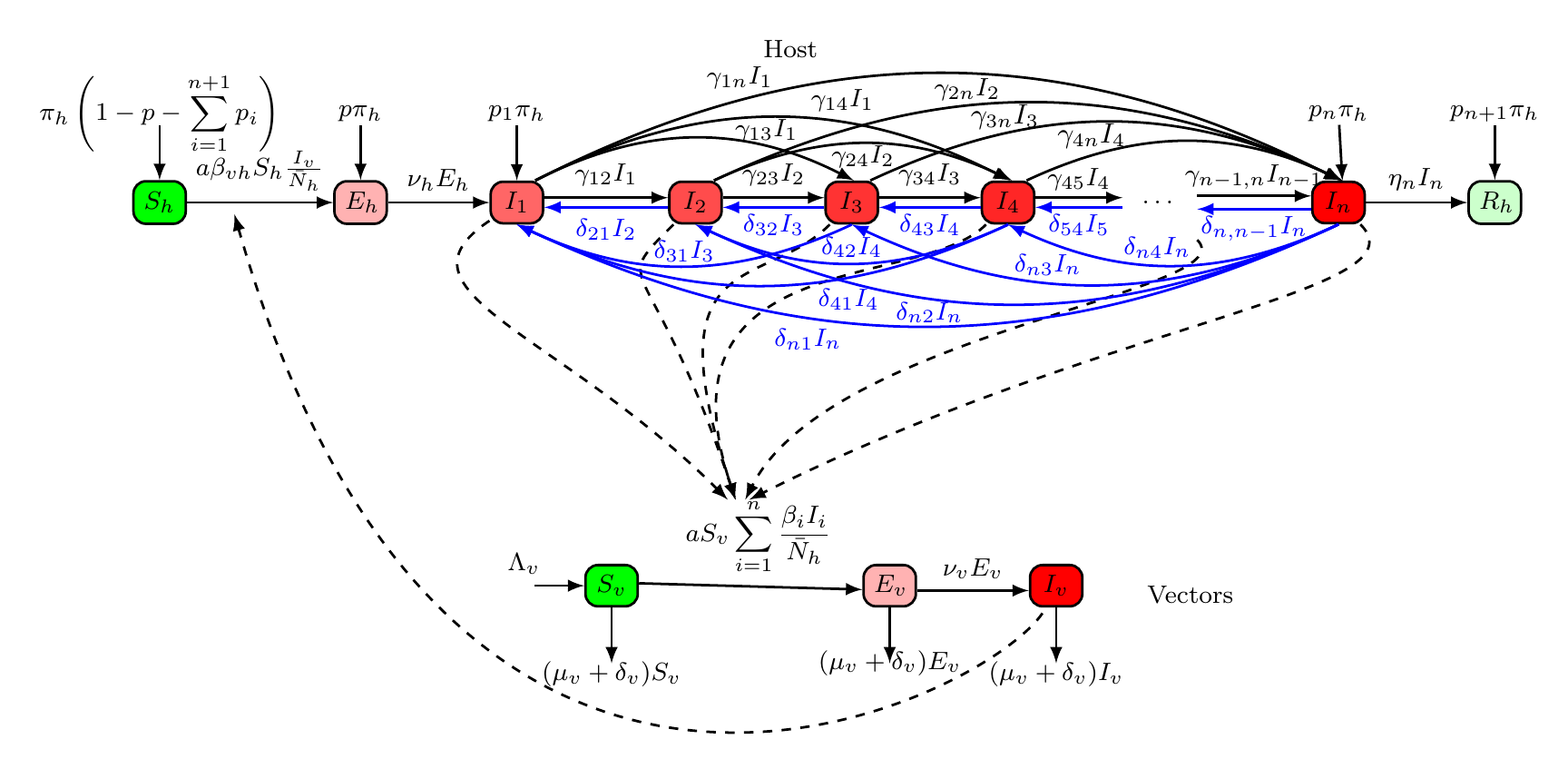}
\caption{Flow diagram of Model \ref{ModelAD}. Note that, to unclutter the figure, we did not display the arrows that represent the recruitments for $I_2$, $I_3$ and $I_4$. Similarly, the arrows representing the death and recovery rates in all host classes are not displayed.}
\label{fig:FlowSEIR1HostAmeliorationDeterioration} 
\end{figure}
To ease the notations, let us denote by 
   $$
m_{ij}=\left\{\begin{array}{lll}
\delta_{ji}\quad\textrm{if}\quad i<j,\\
0\quad\textrm{if}\quad i=j,\\
\gamma_{ji}\quad\textrm{if}\quad i>j,
\end{array}\right.
$$
and 
$\alpha_v=\mu_v+\nu_v+\delta_v$, $\alpha_1=\mu_h+\eta_1 +\sum_{j=2}^n\gamma_{1j}=\mu_h+\eta_1 +\sum_{j=2}^nm_{j1}$ and for $i\geq2$,
\begin{eqnarray}\alpha_i&=&\mu_h+\eta_i +\left(\sum_{j=1,j<i}^n\delta_{ij}+\sum_{j=1,j>i}^n\gamma_{ij}\right)\nonumber\\
&=&\mu_h+\eta_i + \sum_{j=1}^nm_{ji}. \nonumber
\end{eqnarray}
The matrix $M=(m_{ij})_{1\leq i,j\leq n}$ is the transfer matrix between Host's infectious classes and the parameters $\alpha_i$ represent the rates at which infected of stage $i$ leave this stage. %
 The total host population is asymptotically constant. Indeed, its evolution is given by $\dot N_h=\pi_h-\mu_hN_h$ and thus, it is straightforward to show that $\displaystyle\lim_{t\to\infty}N_h=\dfrac{\pi_h}{\mu_h}.$
Moreover, the subsystem describing the dynamics of the host is triangular, and hence we can disregard the dynamics of the recovered host $R_h$. 
Hence, by abusively denoting $\displaystyle\lim_{t\to\infty}N_h$ again by $N_h$ and using the theory of asymptotically autonomous systems for triangular systems \cite{CasThi95,0478.93044},  
System (\ref{ModelAD}) could equivalently be written in a compact form as follows:
\begin{equation} \label{ModelAD2Compact}
\left\{\begin{array}{llll}
\displaystyle\dot S_{h}=\pi_h\left( 1-p-\mathds{1}^T\mathbf{p}-p_{n+1} \right)-a\,\beta_{vh}\,S_{h}\,\dfrac{I_{v}}{N_{h}}-\mu_h S_h\\
\displaystyle\dot E_{h}=p\pi_h+ \,a\beta_{vh}\,S_{h}\,\dfrac{I_{v}}{N_{h}}-(\mu_h+\nu_h+\eta)E_h\\
\displaystyle\dot{\mathbf{I}}_h=\pi_h\mathbf{p}+\nu_hE_he_1-(\textrm{diag}(\alpha)-M)\mathbf{I}_h\\
\displaystyle\dot S_{v}=\pi_v
 -a\dfrac{ S_{v}}{N_h}\bigl\langle
      \Beta   
 \, \mbox{\Large $\mid$} \, \mathbf{I}_h   
\bigr\rangle
-(\mu_v+\delta_v)S_{v}\\
\displaystyle\dot E_{v}=a\dfrac{ S_{v}}{N_h}\bigl\langle
      \Beta    
 \, \mbox{\Large $\mid$} \, \mathbf{I}_h   
\bigr\rangle-(\mu_v+\nu_v+\delta_v)E_{v}\\
\displaystyle\dot I_{v}=\nu_vE_v-(\mu_v+\delta_v)I_{v},
\end{array}\right.
\end{equation}
    where $\mathbf{I}_h=(I_1,I_2,\dots,I_n)^T$, $\Beta=(\beta_1,\beta_2,\dots,\beta_n)^T$, $\mathbf{p}=(p_1,p_2,\dots,p_n)^T$, and $M=G^T+D^T$ with $G=(\gamma_{ij})$ representing the progression matrix, or forward flow transition matrix while $D=(\delta_{ij})$ represents the amelioration matrix, or backward transition flow matrix. More precisely,
  
      \begin{equation}\label{GD}G=\begin{pmatrix}
    0 &   \gamma_{12}&\gamma_{13} & \dots  & \gamma_{1n}\\
  0 &0 &\gamma_{23} & \dots  & \gamma_{2n}\\
    \vdots& \vdots &\ddots  & \ddots &  \vdots \\
    0 & 0 &0 & \dots &  \gamma_{n-1,n}\\ 
0 & 0 &0 & \dots  & 0 
    \end{pmatrix}
   \quad\textrm{and}\quad 
    D=\begin{pmatrix}
    0 &  0&  \dots & 0& 0\\
   \delta_{21} &0 &  \dots & 0& 0\\
    \vdots& \ddots &\ddots   & \vdots& \vdots \\
     \delta_{n-1,1} & \delta_{n-1,2}   & \dots &0& 0\\ 
 \delta_{n1} & \delta_{n2}   & \dots &\delta_{n,n-1}& 0 
    \end{pmatrix}.
    \end{equation}
The parameters of System (\ref{ModelAD}) are described in Table \ref{TableHV1}. The flow chart capturing the infection process is represented in Fig.~(\ref{fig:FlowSEIR1HostAmeliorationDeterioration}). \\

\begin{table}[h!]\label{TableHV1}
  \begin{center}
    \caption{Description of the parameters used in System (\ref{ModelAD}).}
    \label{tab:Param}
    \begin{tabular}{cc}
     \toprule
      Parameters & Description \\
      \midrule
 $\pi_h$  & Recruitment of the host\\
  $\pi_v$  & Recruitment of vectors\\
   $p$  & Proportion of latent immigrants \\
    $p_i$  & Proportion of infectious immigrants at stage $i$\\
 $a$  & Biting rate \\
$\mu_h$  & Host's natural death rate\\
$\beta_{v,h}$  & Host's infectiousness by mosquitoes per biting\\  
$\beta_{i}$ & Vector's infectiousness by host at stage $i$ per biting\\  
$\nu_h$  & Host's rate at which the exposed individuals become infectious\\
$\eta_i$ & Per capita recovery rate of an infected  host at stage $i$\\
$\gamma_{ij}$ &  Host's per capita  progression rate from stage $i$ to $j$\\
$\delta_{ij}$ &  Host's per capita  regression rate from stage $i$ to $j$\\
$\mu_v$  & Vectors' natural mortality rate\\
$\delta_v$  & Vectors' control-induced mortality rate\\
$\nu_v$  & Rate at which the exposed vectors become infectious\\
  \bottomrule
    \end{tabular}
  \end{center}
\end{table}

Model (\ref{ModelAD2Compact}) follows an $SEI^nR-SEI$ structure. That is, of the host and vector populations dynamics follow an and $SEI^nR$ and $SEI$ types of model, respectively. The choices are make to capture some key features in modeling different vector-borne diseases. Indeed, many special cases could be obtained from our general framework to fit a particular arboviral disease. For instance, if $\nu_h\to\infty$, the Host's dynamics will be an $SI^nR$ model. An $SIR-SI$ model have been considered for malaria \cite{tumwiine2010host} and Dengue \cite{Esteva98Dengue} while an $SI^nR-SI$ model was deemed more suited for tick-borne relapsing fever \cite{johnson2016modeling,palmer2018dynamics}.
\\

Model (\ref{ModelAD2Compact}) generalizes other models proposed in the literature in the following five ways:
\begin{itemize}
\item If $D=\mathbf{0}_{n,n}$ and $\gamma_{ij}=0$ for all $i,j$, except when $j=i+1$, Model (\ref{ModelAD2Compact}) consists of a class of staged progression vector-borne diseases models with an influx of infected individuals of each class into the considered population. In this case, 
\begin{itemize}
\item Model (\ref{ModelAD2Compact}) extents the existing stage progression vector-borne models to incorporate a differential proportions of the overall recruitment in all infected classes. This allows us to gauge the impact of imported cases on the dynamics of vector-borne infections. When $\mathbf{p}=\mathbf{0}$, $p=0$, in Model (\ref{ModelAD2Compact}), we obtain the model proposed and studied in \cite{BicharaIggidrSmith2017}. Moreover, our model extends also \cite{palmer2018dynamics}, for which  $\beta_i=\beta$, for all $i$ and the recruitment constitutes of susceptible individuals only. Our model generalizes also \cite{cruz2012control}, which considers Chagas disease model with two stages, namely acute and chronic phases. 
\item  Model (\ref{ModelAD2Compact}) generalizes existing models that investigate staged-progression for directly transmitted infections for which influx of infected individuals are considered \cite{guo2012impacts,1056.92052} and \cite{BraVdd01}, where no stages are considered in the latter.
\item Model (\ref{ModelAD2Compact})  extends also the model in \cite{tumwiine2010host}, where the authors considered a host-vector model $SIR-SI$ with infectious immigrants in investigating the effects of the latter on Malaria dynamics, by incorporating a latent class, $n$ stages of infection in the host's dynamics and a differential infectivity of vectors with respect to host's infectious stages.
 \end{itemize}
 \item If $D$ and $G$ are  as defined in (\ref{GD}), our model extends \cite{BicharaIggidrSmith2017,BraVdd01,cruz2012control,guo2012impacts,1056.92052,palmer2018dynamics,tumwiine2010host} by incorporating forward (deterioration) and backward (amelioration) stage progression. Moreover, the progressions or regressions are not necessarily incremental.
 \item Model  (\ref{ModelAD2Compact}) extends the models considered in \cite{guo2012impacts,guo2012global,1008.92032} to vector-borne disease models and the incorporation of influx of infected in each of the hosts' infectious classes.
 \end{itemize}
 Overall, Model  (\ref{ModelAD2Compact}) generalizes in some fashion or aspect models in \cite{BicharaIggidrSmith2017,BraVdd01,cruz2012control,guo2012impacts,guo2012global,1056.92052,palmer2018dynamics,tumwiine2010host,1008.92032}.

Also, it worthwhile to notice that our system could be seen as vector-borne model that vertically transmitted at each stage of the infection. That is, when off-springs are infected by mothers during pregnancy or delivery. Zika virus is a natural example of a vector-borne disease that is vertically transmitted \cite{tabata2016zika}. \\

The following result shows the solutions of Model (\ref{ModelAD2Compact}) are positive and remain bounded at all times, thereby making the model biologically grounded. 
\begin{lemma} The set
$$\Omega=\left\{ (S_h,E_h,\mathbf{I}_h,S_v,E_v,I_v)\in\R^{n+5}\;\mid S_h+E_h+\mathds{1}^T\mathbf{I}_h\leq N_h, S_v+E_v+I_v\leq N_v:=\frac{\pi_v}{\mu_v+\delta_v} \right\}$$
is a compact positively invariant for System (\ref{ModelAD2Compact}).
\end{lemma}
In the next section, we investigate the steady states solutions of Model (\ref{ModelAD2Compact}) and their asymptotic behavior.

\section{Global Stability Analysis}\label{sec:GSVBAP}  
The next theorem establishes the existence of endemic equilibria of System (\ref{ModelAD2Compact}) and provides conditions under which they may exist. Following Thieme \cite{MR1993355}, we use the nomenclature \textit{strongly} endemic equilibrium if all of its components are positive and \textit{weakly} endemic equilibrium, if at least one of the infected component is positive. Naturally, we start with the assumption that $\mathbf{p}\neq\mathbf{0}_{\mathbb{R}^n}$. The case $\mathbf{p}=\mathbf{0}_{\mathbb{R}^n}$ is dealt in \Cref{SharpThreshold}.
\begin{theorem}\label{ExistenceEEAD}\hfill

The equilibria of System (\ref{ModelAD2Compact})  are as follows:
\begin{enumerate}
\item\label{i} If $\Beta=0$, a unique weakly endemic equilibrium $(\bar{S}_h, \bar{E}_h,\bar{\mathbf{I}}_h,\bar{S}_v,0,0)$ exists.
\item\label{ii} If $\Beta\neq0$ and $p\neq0$, it exists a unique strongly endemic equilibrium $({S}_h^*, {E}_h^*,\mathbf{I}_h^*,S_v^*,E_v^*,I_v^*)$.
\item\label{iii} If $\Beta$ and $\mathbf{p}$ are such that 
$\bigl\langle
      \Beta   
 \, \mbox{\Large $\mid$} \, (\textrm{diag}(\alpha)-M)^{-1} \mathbf{p}  
\bigr\rangle\neq0$, a unique strongly endemic equilibrium $({S}_h^\sharp, E_h^\sharp,\mathbf{I}_h^\sharp,{S}_v^\sharp,E_v^\sharp,I_v^\sharp)$ exists.
\item\label{iv} If \Cref{i}, \Cref{ii} and \Cref{iii} are not satisfied, then a threshold $\mathcal N_0^2(p,\mathbf{p},p_{n+1})$, defined by:
$$
\mathcal N_0^2(p,\mathbf{p},p_{n+1})=\frac{a^2\beta_{vh}\nu_v\nu_h}{\alpha_v(\mu_v+\delta_v)\alpha_h\mu_h}\frac{N_v}{N_h}\frac{\pi_h\left( 1-p-\mathds{1}^T\mathbf{p}-p_{n+1} \right)}{N_h}\bigl\langle
      \Beta   
 \, \mbox{\Large $\mid$} \, (\textrm{diag}(\alpha)-M)^{-1} e_1  
\bigr\rangle,
$$
 exists and for which
\begin{itemize}
\item If $\mathcal N_0^2(p,\mathbf{p},p_{n+1})\leq1$, a unique weakly endemic equilibrium $({S}_h^\diamond, 0,\mathbf{I}_h^\diamond,{S}_v^0,0,0)$ where $\mathbf{I}_h^\diamond>0$ exists. 
\item  If $\mathcal N_0^2(p,\mathbf{p},p_{n+1})>1$, a unique strongly endemic equilibrium $(\tilde{S}_h, \tilde{E}_h,\tilde{\mathbf{I}}_h,\tilde{S}_v,\tilde{E}_v,\tilde{I}_v)$ exists.
\end{itemize}
\end{enumerate}
\end{theorem}
\begin{proof}
An equilibrium $(S_h^\ast,E_h^\ast,I_h^\ast,\dots,I_n^\ast,S_v^\ast,E_v^\ast,I_v^\ast)$ for  Model (\ref{ModelAD2Compact}) satisfies the following relations.
\begin{equation} \label{ModelAD2Endemic1}
\left\{\begin{array}{llll}
\displaystyle \Lambda_h=a\,\beta_{vh}\,S_{h}^\ast\,\dfrac{I_{v}^\ast}{N_{h}}+\mu_h S_h^\ast\\
\displaystyle (\mu_h+\nu_h+\eta)E_h^\ast=p\pi_h+ \,\beta_{vh}\,S_{h}^\ast\,\dfrac{I_{v}^\ast}{N_{h}}\\
\displaystyle (\textrm{diag}(\alpha)-M)\mathbf{I}_h^\ast=\pi_h \mathbf{p}+\nu_hE_h^\ast e_1\\
\displaystyle \pi_v=
 a\dfrac{ S_{v}^\ast}{N_h}\bigl\langle
      \Beta   
 \, \mbox{\Large $\mid$} \, \mathbf{I}_h^\ast   
\bigr\rangle
+(\mu_v+\delta_v)S_{v}^\ast\\
\displaystyle (\mu_v+\nu_v+\delta_v)E_{v}^\ast=a\dfrac{ S_{v}^\ast}{N_h}\bigl\langle
      \Beta    
 \, \mbox{\Large $\mid$} \, \mathbf{I}_h^\ast   
\bigr\rangle\\
\displaystyle (\mu_v+\delta_v)I_{v}^\ast=\nu_vE_v^\ast,
\end{array}\right.
\end{equation}
where $\displaystyle\Lambda_h=\pi_h\left( 1-p-\mathds{1}^t\mathbf{p} -p_{n+1}\right)$.
Using these relationship and $N_v=S_v-E_v-I_v$, one could express $I_v^\ast$ in terms of $\mathbf{I}_h^\ast$, as follows:
\begin{equation}\label{IVstarEE}I_v^\ast\left(\mu_v+\delta_v+\frac{a}{N_h}\bigl\langle
      \Beta   
 \, \mbox{\Large $\mid$} \, \mathbf{I}_h^\ast   
\bigr\rangle\right)=\frac{a\nu_v}{(\mu_v+\nu_v+\delta_v)}\frac{N_v}{N_h}\bigl\langle
      \Beta   
 \, \mbox{\Large $\mid$} \, \mathbf{I}_h^\ast   
\bigr\rangle.\end{equation}
 Moreover, the first equation of (\ref{ModelAD2Endemic1}) leads to: $$S_h^\ast=\frac{\Lambda_hN_h}{\mu_hN_h+a\beta_{vh}I_v^\ast}>0.$$ 
Furthermore, since the matrix $\textrm{diag}(\alpha)-M$ is strictly diagonally dominant and thus invertible, we obtain:
$$\mathbf{I}_h^\ast=\pi_h(\textrm{diag}(\alpha)-M)^{-1}\mathbf{p}+\nu_hE_h^\ast(\textrm{diag}(\alpha)-M)^{-1}e_1,$$
and
\begin{eqnarray*}E_h^\ast 
&=& \frac{1}{(\mu_h+\nu_h+\eta)}\left(p\pi_h+\,\dfrac{\Lambda_ha\,\beta_{vh}I_{v}^\ast}{\mu_hN_h+a\beta_{vh}I_v^\ast}\right).
\end{eqnarray*}
Therefore, $\mathbf{I}_h^\ast$ could be written as:
\begin{equation}\label{Ihast}
\mathbf{I}_h^\ast=\pi_h(\textrm{diag}(\alpha)-M)^{-1}\mathbf{p}+\frac{\nu_h}{(\mu_h+\nu_h+\eta)}\left(p\pi_h+\,\dfrac{\Lambda_ha\,\beta_{vh}I_{v}^\ast}{\mu_hN_h+a\beta_{vh}I_v^\ast}\right)(\textrm{diag}(\alpha)-M)^{-1}e_1.
\end{equation}
The relation (\ref{Ihast}) is key for the remaining of the proof,
as we will use it  to compute  $\bigl\langle
      \Beta   
 \, \mbox{\Large $\mid$} \, \mathbf{I}_h^\ast   
\bigr\rangle$ and obtain a quadratic equation in $I_v^\ast$ using \Cref{IVstarEE}. The latter equation leads to:

\begin{multline}\label{IVstarEEQuadratic}
0=I_v^\ast\left(\mu_v+\delta_v+\frac{a}{N_h}\bigl\langle
      \Beta   \, \mbox{\Large $\mid$} \, \pi_h(\textrm{diag}(\alpha)-M)^{-1}\mathbf{p}\bigr\rangle+\right.\\
      \left.+\frac{a}{N_h}\frac{\nu_h}{(\mu_h+\nu_h+\eta)}\left(p\pi_h+\,\dfrac{\Lambda_ha\,\beta_{vh}I_{v}^\ast}{\mu_hN_h+a\beta_{vh}I_v^\ast}\right)\bigl\langle\Beta   \, \mbox{\Large $\mid$}   (\textrm{diag}(\alpha)-M)^{-1}e_1  
\bigr\rangle\right)-\\
\frac{a\nu_v}{(\mu_v+\nu_v+\delta_v)}\frac{N_v}{N_h}\bigl\langle
      \Beta   
 \, \mbox{\Large $\mid$} \, \pi_h(\textrm{diag}(\alpha)-M)^{-1}\mathbf{p}\bigr\rangle-\\
 \frac{\nu_h}{(\mu_h+\nu_h+\eta)}\left(p\pi_h+\,\dfrac{\Lambda_ha\,\beta_{vh}I_{v}^\ast}{\mu_hN_h+a\beta_{vh}I_v^\ast}\right)\frac{a\nu_v}{(\mu_v+\nu_v+\delta_v)}\frac{N_v}{N_h}\bigl\langle
      \Beta   
 \, \mbox{\Large $\mid$} \,(\textrm{diag}(\alpha)-M)^{-1}e_1  
\bigr\rangle\end{multline}
After some rearrangement, \Cref{IVstarEEQuadratic} could be written as \begin{equation}\label{quad}A{I_v^\ast}^2+BI^\ast_v+C=0,\end{equation} where
\begin{align*}
A&=a\beta_{vh}\left[  \mu_v+\delta_v+\frac{a}{N_h}\bigl\langle
      \Beta   \, \mbox{\Large $\mid$} \, \pi_h(\textrm{diag}(\alpha)-M)^{-1}\mathbf{p}\bigr\rangle+\right.\\
      & \left.+\frac{a}{N_h}\frac{\nu_h}{(\mu_h+\nu_h+\eta)}\left(p\pi_h+ \Lambda_h   \right)\bigl\langle\Beta   \, \mbox{\Large $\mid$}   (\textrm{diag}(\alpha)-M)^{-1}e_1  
\bigr\rangle \right]\\
&>0,
\end{align*}
\begin{align*}
B&=(\mu_v+\delta_v)\mu_nN_h +\frac{a}{N_h}\pi_h\bigl\langle
      \Beta   \, \mbox{\Large $\mid$} \,(\textrm{diag}(\alpha)-M)^{-1}\mathbf{p}\bigr\rangle\mu_hN_h\\
      &+\frac{a}{N_h}\frac{\nu_h}{(\mu_h+\nu_h+\eta)}p\pi_h\bigl\langle\Beta   \, \mbox{\Large $\mid$}   (\textrm{diag}(\alpha)-M)^{-1}e_1  
\bigr\rangle \mu_nN_h\\
&-\frac{a^2\beta_{vh}\nu_v}{(\mu_v+\nu_v+\delta_v)}\frac{N_v}{N_h}\pi_h\bigl\langle
      \Beta   \, \mbox{\Large $\mid$} \,(\textrm{diag}(\alpha)-M)^{-1}\mathbf{p}\bigr\rangle\\
&-\frac{a^2\beta_{vh}\nu_v\nu_h}{(\mu_v+\nu_v+\delta_v)(\mu_h+\nu_h+\eta)}\frac{N_v}{N_h}p\pi_h\bigl\langle
      \Beta   \, \mbox{\Large $\mid$} \,(\textrm{diag}(\alpha)-M)^{-1}e_1\bigr\rangle\\
&-\frac{a^2\beta_{vh}\nu_v\nu_h}{(\mu_v+\nu_v+\delta_v)(\mu_h+\nu_h+\eta)}\frac{N_v}{N_h}\Lambda_h\bigl\langle
      \Beta   \, \mbox{\Large $\mid$} \,(\textrm{diag}(\alpha)-M)^{-1}e_1\bigr\rangle,
\end{align*}
and
\begin{align*}
C&=-\frac{a\nu_v}{(\mu_v+\nu_v+\delta_v)}\frac{N_v}{N_h}\pi_h\mu_hN_h \bigl\langle
      \Beta   
 \, \mbox{\Large $\mid$} \, \pi_h(\textrm{diag}(\alpha)-M)^{-1}\mathbf{p}\bigr\rangle\\
 &- 
 \frac{a\nu_v N_v}{(\mu_v+\nu_v+\delta_v)} \frac{N_v}{N_h} \frac{\nu_h \mu_hN_h}{(\mu_h+\nu_h+\eta)}p\pi_h\bigl\langle
      \Beta   
 \, \mbox{\Large $\mid$} \,(\textrm{diag}(\alpha)-M)^{-1}e_1  
\bigr\rangle.
\end{align*}
Now, we investigate cases for which \Cref{IVstarEEQuadratic} has non-negative solutions. \\

$\bullet$ If $\Beta=0$, then $C=0$ and $B=(\mu_v+\delta_v)\mu_nN_h>0$. Hence, $I_v^*=0$ is the unique solution of the quadratic equation. Thus, the unique equilibrium for System (\ref{ModelAD2Compact}) is
$(\bar{S}_h, \bar{E}_h,\bar{\mathbf{I}}_h,\bar{S}_v,0,0)$, where
$\bar{S}_h=\frac{\Lambda_h}{\mu_h}$, $ \bar E_h=\frac{p\pi_h}{(\mu_h+\nu_h+\eta)}$,
$\bar{\mathbf{I}}_h^\ast=\pi_h(\textrm{diag}(\alpha)-M)^{-1}\mathbf{p}+\frac{\nu_hp\pi_h}{(\mu_h+\nu_h+\eta)} (\textrm{diag}(\alpha)-M)^{-1}e_1,
$
and $\bar{S}_v=\frac{\Lambda_v}{\mu_v+\delta_v}$. This proves \Cref{i}.\\

$\bullet$ If $\Beta\neq0$ and $p\neq0$, then $C<0$ and therefore \Cref{quad} has a unique solution such that $I_v^*>0$. Thus,  from \Cref{Ihast}, System (\ref{ModelAD2Endemic1}), and using the fact that $(\textrm{diag}(\alpha)-M)^{-1}e_1\gg0$, we deduce \Cref{ii}.\\

$\bullet$ If $\Beta$ and $p$ are such that $\bigl\langle
      \Beta   
 \, \mbox{\Large $\mid$} \, (\textrm{diag}(\alpha)-M)^{-1}\mathbf{p}\bigr\rangle\neq0$, then we also have $C<0$; that is, it exists a unique $I_v^\sharp>0$ of \Cref{quad}. As in the previous point, this leads to \Cref{iii}.\\ 

$\bullet$ If the conditions of \Cref{i}, \Cref{ii} and \Cref{iii} are not satisfied; that is, if $\Beta\neq0$, $p=0$ and $\mathbf{p}$ is such that $\bigl\langle
      \Beta   
 \, \mbox{\Large $\mid$} \, (\textrm{diag}(\alpha)-M)^{-1}\mathbf{p}\bigr\rangle=0$. In this case, $C=0$ and $B$ could be written as:
 \begin{align*}
B   &=(\mu_v+\delta_v)\mu_hN_h \\
&-\frac{a^2\beta_{vh}\nu_v\nu_h}{(\mu_v+\nu_v+\delta_v)(\mu_h+\nu_h+\eta)}\frac{N_v}{N_h}\Lambda_h\bigl\langle
      \Beta   \, \mbox{\Large $\mid$} \,(\textrm{diag}(\alpha)-M)^{-1}e_1\bigr\rangle  \\
         &=(\mu_v+\delta_v)\mu_hN_h \left[1
-\frac{a^2\beta_{vh}\nu_v\nu_h}{(\mu_v+\nu_v+\delta_v)(\mu_v+\delta_v)(\mu_h+\nu_h+\eta)\mu_h}\frac{N_v}{N_h}\frac{\Lambda_h}{N_h}\bigl\langle
      \Beta   \, \mbox{\Large $\mid$} \,(\textrm{diag}(\alpha)-M)^{-1}e_1\bigr\rangle   \right]\\
&=(\mu_v+\delta_v)\mu_hN_h\left(1-\mathcal N_0^2(p,\mathbf{p},p_{n+1}) \right).
\end{align*}
Thus, it follows that if $\mathcal N_0^2(p,\mathbf{p},p_{n+1})\leq1$, then $B\geq0$, leading to $I_v^*=0$ and $\mathbf{I}_h^\ast=\pi_h(\textrm{diag}(\alpha)-M)^{-1}\mathbf{p}>\mathbf{0}_{\R^n}$. If $\mathcal N_0^2(p,\mathbf{p},p_{n+1})>1$, then $B<0$ and therefore $I_v^*>0$, leading to a strongly positive equilibrium.
\end{proof}

The following two theorems establish the global stability analysis for the two types of endemic equilibria exhibited in \Cref{ExistenceEEAD}. This gives a complete description of the global asymptotic behavior of System (\ref{ModelAD2Compact}) whenever there is an influx of infected or infectious individuals into the population.

\begin{theorem}\label{GASEEAD}\hfill\\
Let $({S}_h^*, {E}_h^*,\mathbf{I}_h^*,S_v^*,E_v^*,I_v^*)$ be a strongly endemic equilibrium of Model (\ref{ModelAD2Compact}). This equilibrium is globally asymptotically stable whenever it exists.
\end{theorem}
\begin{proof}\hfill\\

Let
consider the following Lyapunov function $  \mathcal V=  \mathcal V_h+  \mathcal V_v$, where
\begin{eqnarray*}
  \mathcal V_h&=&c_0\int_{S_h^\ast}^{S_h}\left(1-\frac{S_h^\ast}{x}\right)dx+c_0\int_{E_{h}^\ast}^{E_h}\left(1-\frac{E_h^\ast}{x}\right)dx+\sum_{i=1}^nc_{i}\int_{I_{i}^\ast}^{I_{i}}\left(1-\frac{I_{i}^\ast}{x}\right)dx,
  \end{eqnarray*}
  and,
  $$\mathcal V_{v}=c_{v}\int_{S_{v}^\ast}^{S_{v}}\left(1-\frac{S_{v}^\ast}{x}\right)dx +c_{v}\int_{E_{v}^\ast}^{E_{v}}\left(1-\frac{E_{v}^\ast}{x}\right)dx+\frac{\mu_v+\nu_v+\delta_v}{\nu_v}c_{v}\int_{I_{v}^\ast}^{I_{v}}\left(1-\frac{I_{v}^\ast}{x}\right)dx.$$
  The coefficients $c=(c_1,c_2,\dots,c_n)^T$ are positive to be determined later. The coefficient $c_0$ and $c_{v}$ are related to $c_1$ as follows:
  \begin{equation}\label{c0cvc1}c_0a\beta_{vh}\,S_{h}^*\dfrac{I_{v}^*}{N_{h}}=c_1\nu_hE_h^*\quad\textrm{and}\quad   c_{v} aS_v^\ast\frac{1}{N_h}=c_1\frac{\nu_hE_h^\ast}{\sum_{i=1}^n \beta_i I_i^\ast}.\end{equation}
    This function is definite positive. We want to prove that its derivative along the trajectories of System (\ref{ModelAD2Compact}) is definite-negative.
   Throughout the proof, we will be using the component-wise endemic relations (\ref{ModelAD2Endemic1}). That is,
    \begin{equation} \label{ModelAD2Endemic}
\left\{\begin{array}{llll} 
\displaystyle \pi_h\left( 1-p-\sum_{i=1}^{n+1}p_i \right)=a\,\beta_{vh}\,S_{h}^\ast\,\dfrac{I_{v}^\ast}{N_{h}}+\mu_h S_h^\ast\\
\displaystyle (\mu_h+\nu_h+\eta)E_h^\ast=p\pi_h+ \,a\beta_{vh}\,S_{h}^\ast\,\dfrac{I_{v}^\ast}{N_{h}}\\
\displaystyle \alpha_1I_1^\ast=p_1\pi_h+\nu_hE_h^\ast+\sum_{j=2}^nm_{1j}I_{j}^\ast\\
\displaystyle \alpha_2I_{2}^\ast=p_2\pi_h+\sum_{j=1}^nm_{2j}I_j^\ast \\ 
\vdots\\
\displaystyle \alpha_iI_i^\ast=p_i\pi_h+ \sum_{j=1}^nm_{ij}I_j^\ast\\ 
\displaystyle \alpha_nI_n^\ast=p_n\pi_h+\sum_{j=1}^{n}m_{nj}I_{j}^\ast\\
\displaystyle \Lambda_v=a S_{v}^\ast\sum_{i=1}^n\dfrac{\beta_i I_{i}^\ast}{N_{h}}+(\mu_v+\delta_v)S_{v}^\ast\\
\displaystyle (\mu_v+\nu_v+\delta_v)E_{v}^\ast=a S_{v}^\ast\sum_{i=1}^n\dfrac{\beta_i I_{i}^\ast}{N_{h}}\\
\displaystyle (\mu_v+\delta_v)I_{v}^\ast=\nu_vE_v^\ast
\end{array}\right.
\end{equation}

  The  derivative of $\mathcal{V}_h$ along the trajectories of System (\ref{ModelAD2Compact}) is:
  \begin{eqnarray}\label{Lyapc1}
  \dot {\mathcal V_h} &=&c_0\left(1-\frac{S_h^\ast}{S_h}\right)\dot{S}_h+c_0\left(1-\frac{E_h^\ast}{E_{h}}\right)\dot E_{h}+\sum_{i=1}^nc_i\left(1-\frac{I_i^\ast}{I_i}\right)\dot I_i\nonumber\\
&=&c_0\mu_h S_h^\ast\left(2-\frac{S_h}{S_h^\ast}-\frac{S_h^\ast}{S_h}\right)+c_0a\beta_{vh}S_{h}^*\,\dfrac{I_{v}^*}{N_{h}}\left(2 -\frac{S_h^\ast}{S_h}-\frac{S_h}{S_h^\ast} \frac{I_v}{I_v^\ast} \frac{E_h^\ast}{E_{h}}\right)+c_0p\pi\left(2-\frac{E_h^\ast}{E_{h}}-\frac{E_{h}}{E_h^*}  \right)\nonumber\\
  &&-c_0\left(  \,\beta_{vh}\,S_{h}^*\dfrac{I_{v}^*}{N_{h}} \right)\frac{E_{h}}{E_h^*}+c_0a\,\beta_{vh}S_h^\ast\,\dfrac{I_{v}}{N_{h}}+\sum_{i=1}^nc_i\left(1-\frac{I_i^\ast}{I_i}\right)\dot I_i 
   \end{eqnarray}
    Using the endemic relation $\displaystyle\alpha_1I_1^*=p_1\pi_h+\nu_hE_h^*+\sum_{j=1}^nm_{1j}I_{j}^*$, and the relationship between $c_0$ and $c_1$,  Equation (\ref{Lyapc1}) yields to
  \begin{eqnarray}\label{Lyapc111}
  \dot{\mathcal V}_h
 &=&c_0\mu_h S_h^\ast\left(2-\frac{S_h}{S_h^\ast}-\frac{S_h^\ast}{S_h}\right)+c_0a\beta_{vh}S_{h}^*\,\dfrac{I_{v}^*}{N_{h}}\left(3 -\frac{S_h^\ast}{S_h}-\frac{S_h}{S_h^\ast} \frac{I_v}{I_v^\ast} \frac{E_h^\ast}{E_{h}}  -\frac{E_h}{E_h^*}\frac{I_1^\ast}{I_1}\right)+c_0p\pi\left(2-\frac{E_h^\ast}{E_{h}}-\frac{E_{h}}{E_h^*}  \right)\nonumber\\
  &&+c_0a\,\beta_{vh}S_h^\ast\,\dfrac{I_{v}}{N_{h}}+ c_1p_1\pi_h\left(2-\frac{I_1^\ast}{I_1} -\frac{I_1}{I_1^\ast}\right)    -c_1\left(\nu_hE_h^*+\sum_{j=1}^nm_{1j}I_{j}^*\right)\frac{I_1}{I_1^*}+c_1\sum_{j=1}^nm_{1j}I_{j} \nonumber\\
     && +c_1\sum_{j=1}^n m_{1j}I_{j}^*
-c_1\frac{I_1^\ast}{I_1}(\sum_{j=1}^nm_{1j}I_{j}) +\sum_{i=2}^nc_i\left(1-\frac{I_i^\ast}{I_i}\right)\dot I_i 
   \end{eqnarray}
Noting that, from the endemic relations (\ref{ModelAD2Endemic}), we have $\alpha_{i}I_i^\ast=p_i\pi_h+ \sum_{j=1}^nm_{ij}I_j^\ast$, and thus, the last term of Equation (\ref{Lyapc111}) leads to
     \begin{eqnarray}\label{ciIi2}
     c_i\left(1-\frac{I_i^\ast}{I_i}\right)\dot I_i 
       & =&c_ip_i\pi_h\left(2-\frac{I_i^\ast}{I_i} -\frac{I_i}{I_i^*}\right)+ c_i\sum_{j=1}^nm_{ij}I_j^\ast-c_i\left(\sum_{j=1}^nm_{ij}I_j^\ast\right)\frac{I_i}{I_i^*}\nonumber\\
       &&+ c_i\sum_{j=1}^n m_{ij}I_j -c_i\frac{I_i^\ast}{I_i}\left(\sum_{j=1}^n m_{ij}I_j \right)
        \end{eqnarray}

        Moreover, we can check that the derivative of $\mathcal V_v$ along the trajectories of System (\ref{ModelAD2Compact}) is:
        \begin{equation}\label{VvDot1}
         \dot{\mathcal V_v}=c_{v,1}\left( \mathcal A_v+\sum_{i=1}^naS_v^\ast\frac{\beta_i I_i^\ast}{N_i}\left(3  -\frac{S_v^\ast}{S_v}-\frac{S_v}{S_v^\ast} \frac{I_i}{I_i^\ast} \frac{E_v^\ast}{E_{v}}-\frac{E_vI_v^\ast}{E_v^\ast I_v}  \right)+aS_v^\ast\sum_{i=1}^n \frac{\beta_i I_i}{N_h}-aS_v^\ast\sum_{i=1}^n \frac{\beta_i I_i^\ast}{N_h}\frac{I_v}{I_v^\ast}\right),
        \end{equation}
where $\displaystyle\mathcal A_v=(\mu_v+\delta_v)S_v^*\left(2-\dfrac{S_v^*}{S_v}-\dfrac{S_v}{S_v^*}\right).$       
       Combining equations (\ref{Lyapc111}), (\ref{ciIi2}), and (\ref{VvDot1}), we obtain:
        \begin{eqnarray}\label{CombinedLyap}
 \dot{\mathcal V} 
        &=&c_0\mu_h S_h^\ast\left(2-\frac{S_h}{S_h^\ast}-\frac{S_h^\ast}{S_h}\right)+c_0a\beta_{vh}S_{h}^*\,\dfrac{I_{v}^*}{N_{h}}\left(3 -\frac{S_h^\ast}{S_h}-\frac{S_h}{S_h^\ast} \frac{I_v}{I_v^\ast} \frac{E_h^\ast}{E_{h}}  -\frac{E_h}{E_h^*}\frac{I_1^\ast}{I_1}\right)+c_0p\pi\left(2-\frac{E_h^\ast}{E_{h}}-\frac{E_{h}}{E_h^*}  \right)\nonumber\\
  &&+c_0a\,\beta_{vh}S_h^\ast\,\dfrac{I_{v}}{N_{h}}+\sum_{i=1}^nc_ip_i\pi_h\left(2-\frac{I_i^\ast}{I_i} -\frac{I_i}{I_i^*}\right)   -  c_1\left(\nu_hE_h^*+\sum_{j=1}^nm_{1j}I_{j}^*\right)\frac{I_1}{I_1^*}+c_1\sum_{j=1}^nm_{1j}I_{j} \nonumber\\
     && +c_1\sum_{j=1}^n m_{1j}I_{j}^*
-c_1\frac{I_1^\ast}{I_1}\left(\sum_{j=1}^nm_{1j}I_{j}\right) \nonumber\\
&&+\sum_{i=2}^n\left[
 c_i\sum_{j=1}^nm_{ij}I_j^\ast-c_i\left(\sum_{j=1}^nm_{ij}I_j^\ast\right)\frac{I_i}{I_i^*} 
      + c_i\sum_{j=1}^n m_{ij}I_j -c_i\frac{I_i^\ast}{I_i}\left(\sum_{j=1}^n m_{ij}I_j \right)
\right]\nonumber\\
%
%
&&+c_{v,1}\left( \mathcal A_v+\sum_{i=1}^naS_v^\ast\frac{\beta_i I_i^\ast}{N_i}\left(3  -\frac{S_v^\ast}{S_v}-\frac{S_v}{S_v^\ast} \frac{I_i}{I_i^\ast} \frac{E_v^\ast}{E_{v}}-\frac{E_vI_v^\ast}{E_v^\ast I_v}  \right)+aS_v^\ast\sum_{i=1}^n \frac{\beta_i I_i}{N_h}-aS_v^\ast\sum_{i=1}^n \frac{\beta_i I_i^\ast}{N_h}\frac{I_v}{I_v^\ast}\right)
   \end{eqnarray}
    Given the relationship (\ref{c0cvc1}), the linear terms in $I_v$ in Equation (\ref{CombinedLyap}) cancel. Furthermore, by substituting $\mathcal A_v$ by its expression and  $c_{v,1}$ by their expressions, Equation (\ref{CombinedLyap}) leads to
          \begin{eqnarray}\label{CombinedLyap3}
\dot{\mathcal V} 
        &=&c_0\mu_h S_h^\ast\left(2-\frac{S_h}{S_h^\ast}-\frac{S_h^\ast}{S_h}\right)+c_{v,1}(\mu_v+\delta_v)S_v^*\left(2-\dfrac{S_v^*}{S_v}-\dfrac{S_v}{S_v^*}\right)\nonumber\\
        &&+c_0p\pi\left(2-\frac{E_h^\ast}{E_{h}}-\frac{E_{h}}{E_h^*}  \right)+\sum_{i=1}^nc_ip_i\pi_h\left(2-\frac{I_i^\ast}{I_i} -\frac{I_i}{I_i^*}\right) \nonumber\\
        &&+c_{v,1}\sum_{i=1}^naS_v^\ast\frac{\beta_i I_i^\ast}{N_i}\left(6 -\frac{S_h^\ast}{S_h}-\frac{S_h}{S_h^\ast} \frac{I_v}{I_v^\ast} \frac{E_h^\ast}{E_{h}}  -\frac{E_h}{E_h^*}\frac{I_1^\ast}{I_1}  -\frac{S_v^\ast}{S_v}-\frac{S_v}{S_v^\ast} \frac{I_i}{I_i^\ast} \frac{E_v^\ast}{E_{v}}-\frac{E_vI_v^\ast}{E_v^\ast I_v} \right)\nonumber\\
  &&  -  c_1\left(\nu_hE_h^*+\sum_{j=1}^nm_{1j}I_{j}^*\right)\frac{I_1}{I_1^*}+c_1\sum_{j=1}^nm_{1j}I_{j} 
   +c_1\sum_{j=1}^n m_{1j}I_{j}^*
-c_1\frac{I_1^\ast}{I_1}\left(\sum_{j=1}^nm_{1j}I_{j}\right) \nonumber\\
&&+\sum_{i=2}^n\left[
 c_i\sum_{j=1}^nm_{ij}I_j^\ast-c_i\left(\sum_{j=1}^nm_{ij}I_j^\ast\right)\frac{I_i}{I_i^*} 
      + c_i\sum_{j=1}^n m_{ij}I_j -c_i\frac{I_i^\ast}{I_i}\left(\sum_{j=1}^n m_{ij}I_j \right)
\right]\nonumber\\
%
%
&&+c_1\frac{\nu_hE_h^\ast}{\sum_{l=1}^n \beta_l I_l^\ast}\left(\sum_{i=1}^n  \beta_i I_i\right)
   \end{eqnarray}
 We choose the vector $c=(c_1,c_2,\dots,c_n)^T$ to be the solution of  the linear system $Bc=0$, where
 \begin{equation}\label{B}
B=\begin{pmatrix}
 -\diamond_{11} & m_{21}I_1^\ast &  m_{31}I_{1}^\ast & \dots & \dots & m_{n1}I_{1}^\ast\\
\nu_hE_h^*\frac{\beta_2 I_2^\ast}{\sum_{i=1}^n \beta_i I_i^\ast}+ m_{12}I_{2}^\ast  & -\diamond_{22} & m_{32}I_2^\ast & \dots &\dots & m_{n2}I_{2}^\ast\\
\nu_hE_h^*\frac{\beta_3 I_3^\ast}{\sum_{i=1}^n \beta_i I_i^\ast} + m_{13}I_{3}^\ast & m_{23}I_{3}^\ast & -\diamond_{33} & \dots& \dots&  m_{n3}I_{n}^\ast\\
 \vdots& \vdots & \vdots& \ddots & \ddots& \vdots \\
\nu_hE_h^* \frac{\beta_n I_n^\ast}{\sum_{i=1}^n \beta_i I_i^\ast}+m_{1n}I_{n}^\ast & m_{2n}I_{n}^\ast& m_{3n}I_{n}^\ast & \hdots&\hdots &-\diamond_{nn}
\end{pmatrix}
\end{equation}
where $$\diamond_{11}=\left(\nu_hE_h^*\frac{\sum_{i=2}^n \beta_i I_i^\ast }{\sum_{i=1}^n \beta_i I_i^\ast }+\sum_{j=1}^nm_{1j}I_{j}^*\right),\quad\textrm{
and for   } k\geq 2,\quad
 \diamond_{kk}=\sum_{j=1}^{n}m_{kj}I_j^\ast.$$
   The matrix $B$ is irreducible. 
     Indeed, since  $\mathbf{I}_h^*\gg0$, we notice that all elements of the second upper diagonal of $B$ are all non zero, as $m_{i+1,i}=\gamma_{i,i+1}$, and thus represent the incremental transition between infectious classes. This, along with the first column, makes the matrix $B$ irreducible. 
 Hence, it could be shown that $\dim(\ker(B))= 1$; and by the Kirchhoff's matrix tree theorem\cite{bollobas2013modern,moon1970counting}, $c_i=-C_{ii}\gg0$ where $C_{ii}$ is the cofactor of the $i^{th}$ diagonal of  
$B$. Hence, it exists $c=(c_1,c_2,\dots,c_n)^T\gg0$ such that $Bc=0$. Moreover, this implies that, in Equation (\ref{CombinedLyap3}), we have:
\begin{eqnarray*}
0&=&  -  c_1\left(\nu_hE_h^*+\sum_{j=1}^nm_{1j}I_{j}^*\right)\frac{I_1}{I_1^*}+c_1\sum_{j=1}^nm_{1j}I_{j} 
   +c_1\sum_{j=1}^n m_{1j}I_{j}^*
 \nonumber\\
&&+\sum_{i=2}^n\left[
 c_i\sum_{j=1}^nm_{ij}I_j^\ast-c_i\left(\sum_{j=1}^nm_{ij}I_j^\ast\right)\frac{I_i}{I_i^*} 
      + c_i\sum_{j=1}^n m_{ij}I_j  
\right]\nonumber\\
%
%
&&+c_1\frac{\nu_hE_h^\ast}{\sum_{l=1}^n \beta_l I_l^\ast}\left(\sum_{i=1}^n  \beta_i I_i\right)
   \end{eqnarray*}
Thus, (\ref{CombinedLyap3}) yields to:

 \begin{align}\label{CombinedLyap4}
\dot{\mathcal V} 
 &= c_0\mu_h S_h^\ast\left(2-\frac{S_h}{S_h^\ast}-\frac{S_h^\ast}{S_h}\right)+c_{v,1}(\mu_v+\delta_v)S_v^*\left(2-\dfrac{S_v^*}{S_v}-\dfrac{S_v}{S_v^*}\right)\nonumber\\
 & +c_0p\pi\left(2-\frac{E_h^\ast}{E_{h}}-\frac{E_{h}}{E_h^*}  \right)+\sum_{i=1}^nc_ip_i\pi_h\left(2-\frac{I_i^\ast}{I_i} -\frac{I_i}{I_i^*}\right) \nonumber\\
        & +c_{v,1}\sum_{i=1}^naS_v^\ast\frac{\beta_i I_i^\ast}{N_i}\underbrace{\left(6 -\frac{S_h^\ast}{S_h}-\frac{S_h}{S_h^\ast} \frac{I_v}{I_v^\ast} \frac{E_h^\ast}{E_{h}}  -\frac{E_h}{E_h^*}\frac{I_1^\ast}{I_1}  -\frac{S_v^\ast}{S_v}-\frac{S_v}{S_v^\ast} \frac{I_i}{I_i^\ast} \frac{E_v^\ast}{E_{v}}-\frac{E_vI_v^\ast}{E_v^\ast I_v} \right)}_{\mathcal W_i}\nonumber\\
 & +\sum_{i=1}^n\sum_{j=1}^nc_im_{ij}I_j^\ast \left(
1
       -\frac{I_i^\ast}{I_i} \frac{I_j^\ast}{I_j}  
\right).
   \end{align}

 The first three terms of (\ref{CombinedLyap4}) are definite-positive. Now, we will break down the last two terms in (\ref{CombinedLyap4}) into definite-negative terms. Indeed, following \cite{guo2012global}, we transform each theses expressions as sums of terms in the form of $f(x)=1-x+\ln x$. To this end, we will use the fact that the function $f(x)$ is definite negative around $x^\ast=1$. Indeed, using the properties of natural logarithm function, the expression of $\mathcal W_i$ in (\ref{CombinedLyap4}) could be written as:
  \begin{align*}
\mathcal W_i &= 6 -\frac{S_h^\ast}{S_h}-\frac{S_h}{S_h^\ast} \frac{I_v}{I_v^\ast} \frac{E_h^\ast}{E_{h}}  -\frac{E_h}{E_h^*}\frac{I_1^\ast}{I_1}  -\frac{S_v^\ast}{S_v}-\frac{S_v}{S_v^\ast} \frac{I_i}{I_i^\ast} \frac{E_v^\ast}{E_{v}}-\frac{E_vI_v^\ast}{E_v^\ast I_v}\nonumber\\
    &= \left(1-\frac{S_h^\ast}{S_h}+\ln\frac{S_h^\ast}{S_h}\right)+\left(1-\frac{S_h}{S_h^\ast} \frac{I_v}{I_v^\ast} \frac{E_h^\ast}{E_{h}}+\ln\frac{S_hI_vE_h^\ast}{S_h^\ast I_v^\ast E_{h}} \right)+\left(1-\frac{E_vI_v^\ast}{E_v^\ast I_v} +\ln\frac{E_vI_v^\ast}{E_v^\ast I_v} \right)\nonumber\\
&   + \left(1 -\frac{E_h}{E_h^\ast}\frac{I_1^\ast}{I_1} +\ln\frac{E_hI_1^\ast}{E_h^\ast I_1}   \right)
+\left(1-\frac{S_vI_i E_v^\ast}{S_v^\ast I_i^\ast E_{v}} +\ln\frac{S_vI_i E_v^\ast}{S_v^\ast I_i^\ast E_{v}} \right)
+\left(1 
  -\frac{S_v^\ast}{S_v}+\ln\frac{S_v^\ast }{S_v}\right)
   +\ln\frac{ I_1}{I_1^\ast}\frac{ I_i^\ast }{I_i} 
      \end{align*} 
   Noting that   $$1- \frac{I_i^\ast}{I_i} \frac{I_j}{I_j^\ast}=1- \frac{I_i^\ast}{I_i} \frac{I_j}{I_j^\ast}+\ln\frac{I_i^\ast I_j}{I_iI_j^\ast} +\ln\frac{I_iI_j^\ast}{I_i^\ast I_j},$$  
and substitute the expression of $\mathcal W_i$,  \Cref{CombinedLyap4} becomes
      \begin{align}\label{CombinedLyapGene2}
 \dot {\mathcal V}
 &= c_0\mu_h S_h^\ast\left(2-\frac{S_h}{S_h^\ast}-\frac{S_h^\ast}{S_h}\right)+c_{v,1}(\mu_v+\delta_v)S_v^*\left(2-\dfrac{S_v^*}{S_v}-\dfrac{S_v}{S_v^*}\right) \nonumber\\
 &+c_{v,1}aS_v^\ast\sum_{i=1}^n \frac{\beta_i I_i^\ast}{N_h}\left[
 \left(1-\frac{S_h^\ast}{S_h}+\ln\frac{S_h^\ast}{S_h}\right)+\left(1-\frac{S_h}{S_h^\ast} \frac{I_v}{I_v^\ast} \frac{E_h^\ast}{E_{h}}+\ln\frac{S_hI_vE_h^\ast}{S_h^\ast I_v^\ast E_{h}} \right)\right.\nonumber\\
 & \left.+\left(1-\frac{E_vI_v^\ast}{E_v^\ast I_v} +\ln\frac{E_vI_v^\ast}{E_v^\ast I_v} \right)
 + \left(1 -\frac{E_h}{E_h^\ast}\frac{I_1^\ast}{I_1} +\ln\frac{E_hI_1^\ast}{E_h^\ast I_1}   \right)
 +\left(1-\frac{S_vI_i E_v^\ast}{S_v^\ast I_i^\ast E_{v}} +\ln\frac{S_vI_i E_v^\ast}{S_v^\ast I_i^\ast E_{v}} \right)
 \right.\nonumber\\
  & \left.+\left(1 
  -\frac{S_v^\ast}{S_v}+\ln\frac{S_v^\ast }{S_v}\right)\right]+c_0p\pi\left(2-\frac{E_h^\ast}{E_{h}}-\frac{E_{h}}{E_h^*}  \right)
  +\sum_{i=1}^nc_ip_i\pi_h\left(2-\frac{I_i^\ast}{I_i}-\frac{I_i}{I_i^\ast}\right) \nonumber\\
  &  \sum_{i=1}^nc_i \sum_{j=1}^{n}m_{ij}I_j^\ast\left(1- \frac{I_i^\ast}{I_i} \frac{I_j}{I_j^\ast}+\ln\frac{I_i^\ast I_j}{I_iI_j^\ast}\right)\nonumber\\
 &+\underbrace{c_{v,1}aS_v^\ast\sum_{i=1}^n \frac{\beta_i I_i^\ast}{N_h} \ln\frac{ I_1}{I_1^\ast}\frac{ I_i^\ast }{I_i}   
 +\sum_{i=1}^nc_i \sum_{j=1}^{n}m_{ij}I_j^\ast \ln\frac{I_iI_j^\ast}{I_i^\ast I_j}}_{\mathcal S} 
      \end{align}

   All but the last two sums in  (\ref{CombinedLyapGene2}) are definite negative. Let us denote by $\mathcal S$ the sum of these two sums. We focus on proving that $\mathcal S:=0$. Indeed, recall the expression of $c_{v,1}$ in terms of $c_1$, given in (\ref{c0cvc1}):
$$
 c_{v,1}\frac{aS_v^\ast}{N_h}=c_1\frac{\nu_hE_h^*}{\sum_{l=1}^n  \beta_l I_l^\ast }.
$$
  By replacing $c_{v,1}$ by its value in $\mathcal S$, we obtain,

 \begin{align}\label{S}
\mathcal S &= c_{v,1}aS_v^\ast\sum_{i=1}^n \frac{\beta_i I_i^\ast}{N_h} \ln\frac{ I_1}{I_1^\ast}\frac{ I_i^\ast }{I_i}   
 +\sum_{i=1}^nc_i \sum_{j=1}^{n}m_{ij}I_j^\ast \ln\frac{I_iI_j^\ast}{I_i^\ast I_j}\nonumber\\
 & =c_1\frac{\nu_hE_h^*}{\sum_{l=1}^n  \beta_l I_l^\ast }\sum_{i=1}^n  \beta_i I_i^\ast \ln\frac{ I_1}{I_1^\ast}\frac{ I_i^\ast }{I_i}   
 +\sum_{i=1}^nc_i \sum_{j=1}^{n}m_{ij}I_j^\ast \ln\frac{I_iI_j^\ast}{I_i^\ast I_j}\nonumber\\
  & =c_1\sum_{j=1}^{n}\ln\frac{I_1I_j^\ast}{I_1^\ast I_j}\left[\frac{\nu_hE_h^*}{\sum_{l=1}^n  \beta_l I_l^\ast } \beta_j I_j^\ast    
 + m_{1j}I_j^\ast \right]
 %
 +\sum_{i=2}^nc_i \sum_{j=1}^{n}m_{ij}I_j^\ast \ln\frac{I_iI_j^\ast}{I_i^\ast I_j} 
 \end{align}
However, since $c_i$ are the components of the solution of $Bc=0$ where $B$ is given in (\ref{B}), it follows that, for any $j\geq2,$

$$
   c_1\left( \nu_hE_h^*\frac{\beta_j I_j^\ast}{\sum_{l=1}^n \beta_l I_l^\ast} +m_{1j}I_{j}^\ast\right)=
   c_j\left(\sum_{k=1}^ n m_{jk}I_k^\ast \right)
   -\sum_{i=2}^{n}c_im_{ij}I_j^\ast,
  $$ 
Plugging this expression into \Cref{S}, and using again the properties of natural logarithms, we obtain:

\begin{align}\label{S2}
\mathcal S
& =\sum_{j=1}^{n}\ln\frac{I_1I_j^\ast}{I_1^\ast I_j}\left[c_j\left(\sum_{k=1}^ n m_{jk}I_k^\ast \right)
   -\sum_{i=2}^{n}c_im_{ij}I_j^\ast \right]
 +\sum_{i=2}^nc_i \sum_{j=1}^{n}m_{ij}I_j^\ast \ln\frac{I_iI_j^\ast}{I_i^\ast I_j}\nonumber\\
  &= \sum_{i=1}^{n}c_i\ln\frac{I_1I_i^\ast}{I_1^\ast I_i}\left(\sum_{k=1}^ n m_{ik}I_k^\ast \right)+\sum_{i=2}^nc_i \sum_{j=1}^{n}m_{ij}I_j^\ast \left[ -\ln\frac{ I_1 I_j^\ast}{I_1^\ast I_j}+\ln\frac{I_i I_j^\ast} {I_i^\ast I_j}\right]\nonumber\\
  &= \sum_{i=2}^{n}c_i\ln\frac{I_1I_i^\ast}{I_1^\ast I_i}\left(\sum_{j=1}^ n m_{ij}I_j^\ast \right)+\sum_{i=2}^nc_i \sum_{j=1}^{n}m_{ij}I_j^\ast \left[ \ln\frac{I_1^\ast I_i}{ I_1 I_i^\ast}\right]\nonumber\\
  &:=0, \end{align}
since for $i=1$, the coefficient of the sum is $\ln1=0$. Finally using \Cref{CombinedLyapGene2} and \Cref{S2}, the derivative of $\mathcal V$ along the trajectories of \Cref{ModelAD2Compact} is
   \begin{align}\label{CombinedLyapGeneFinal}
 \dot {\mathcal V}
 &= c_0\mu_h S_h^\ast\left(2-\frac{S_h}{S_h^\ast}-\frac{S_h^\ast}{S_h}\right)+c_{v,1}(\mu_v+\delta_v)S_v^*\left(2-\dfrac{S_v^*}{S_v}-\dfrac{S_v}{S_v^*}\right) \nonumber\\
 &+c_{v,1}aS_v^\ast\sum_{i=1}^n \frac{\beta_i I_i^\ast}{N_h}\left[
 \left(1-\frac{S_h^\ast}{S_h}+\ln\frac{S_h^\ast}{S_h}\right)+\left(1-\frac{S_h}{S_h^\ast} \frac{I_v}{I_v^\ast} \frac{E_h^\ast}{E_{h}}+\ln\frac{S_hI_vE_h^\ast}{S_h^\ast I_v^\ast E_{h}} \right)\right.\nonumber\\
 & \left.+\left(1-\frac{E_vI_v^\ast}{E_v^\ast I_v} +\ln\frac{E_vI_v^\ast}{E_v^\ast I_v} \right)
 + \left(1 -\frac{E_h}{E_h^\ast}\frac{I_1^\ast}{I_1} +\ln\frac{E_hI_1^\ast}{E_h^\ast I_1}   \right)
 +\left(1-\frac{S_vI_i E_v^\ast}{S_v^\ast I_i^\ast E_{v}} +\ln\frac{S_vI_i E_v^\ast}{S_v^\ast I_i^\ast E_{v}} \right)
 \right.\nonumber\\
  & \left.+\left(1 
  -\frac{S_v^\ast}{S_v}+\ln\frac{S_v^\ast }{S_v}\right)\right]+c_0p\pi\left(2-\frac{E_h^\ast}{E_{h}}-\frac{E_{h}}{E_h^*}  \right)
  +\sum_{i=1}^nc_ip_i\pi_h\left(2-\frac{I_i^\ast}{I_i}-\frac{I_i}{I_i^\ast}\right) \nonumber\\
  &  \sum_{i=1}^nc_i \sum_{j=1}^{n}m_{ij}I_j^\ast\left(1- \frac{I_i^\ast}{I_i} \frac{I_j}{I_j^\ast}+\ln\frac{I_i^\ast I_j}{I_iI_j^\ast}\right),
      \end{align}
which is definite-negative. Therefore, by Lyapunov's stability theorem, the unique endemic equilibrium is GAS.

\end{proof}

\begin{theorem}\label{WEEGAS}
Let $({S}_h^\diamond, 0,\mathbf{I}_h^\diamond,{S}_v^0,0,0)$ be a weakly endemic equilibrium of Model (\ref{ModelAD2Compact}). This equilibrium is GAS.
\end{theorem}

\begin{proof}
Let $({S}_h^\diamond, 0,\mathbf{I}_h^\diamond,{S}_v^0,0,0)$ be a weakly endemic equilibrium of Model (\ref{ModelAD2Compact}). \\

$\bullet$
 If $\Beta=0$, we remark from the vector's equations in Model (\ref{ModelAD2Compact}) that $S_v\to S_v^0:=\frac{\pi_v}{\mu_v+\delta_v}$, $E_v\to0$ and $I_v\to0$ as $t\to\infty$. So, by the theory of asymptotically autonomous systems for triangular systems \cite{CasThi95,0478.93044}, Model (\ref{ModelAD2Compact}) is equivalent to
\begin{equation} 
\left\{\begin{array}{llll}\label{WeakEE}
\displaystyle\dot S_{h}=\pi_h\left( 1-p-\sum_{i=1}^{n+1}p_i \right)-\mu_h S_h\\
\displaystyle\dot E_{h}=p\pi_h-(\mu_h+\nu_h+\eta)E_h\\
\displaystyle\dot{\mathbf{I}}_h=\pi_h\mathbf{p}+\nu_hE_he_1-(\textrm{diag}(\alpha)-M)\mathbf{I}_h
\end{array}\right.
\end{equation}
System (\ref{WeakEE}) is triangular and linear, and its solutions converge toward $(\bar S_h,\bar E_h,\bar{\mathbf{I}}_h)$, where $\bar S_h:=\frac{\Lambda_h}{\mu_h}$, $\bar E_h:=\frac{p\pi_h}{\mu_h+\nu_h+\eta}$ and $\bar{\mathbf{I}}_h:=(\textrm{diag}(\alpha)-M)^{-1}\left( \pi_h\mathbf{p}+\frac{\nu_hp\pi_h}{\mu_h+\nu_h+\eta}e_1 \right)$. Thus, it follows that the weak endemic equilibrium $(\bar S_h,\bar E_h,\bar{\mathbf{I}}_h,S_v^0,0,0)$ of Model (\ref{ModelAD2Compact}) is GAS.\\

Before we start the proof of the next case, let us define the order relation for the vectors as follows: $u\leq v$ if $u_i\leq v_i$, for all $i$, where  $u_i$ and $v_i$ are components of $u$ and $v$ respectively. Similarly, $u< v$ if $u\leq v$ and $u\neq v$. Also $u\gg v$ if $u_i>v_i$, for all $i$.\\

$\bullet$ If \Cref{iv} of \Cref{ExistenceEEAD} is satisfied with $\mathcal N_0^2(p,\mathbf{p},p_{n+1})\leq1$. That is,  
$\Beta\neq0$, $p=0$ and $\mathbf{p}$ and $M$ are such that $\bigl\langle
      \Beta   
 \, \mbox{\Large $\mid$} \, (\textrm{diag}(\alpha)-M)^{-1}\mathbf{p}\bigr\rangle=0$.
These imply that, using the endemic relations, \[(\textrm{diag}(\alpha)-M)^{-1} \mathbf{p} >0.
\]
Moreover, it follows that $\bigl\langle
      \Beta   
 \, \mbox{\Large $\mid$} \, \mathbf{I}_h^\diamond  
\bigr\rangle=0$. This implies that it exists a subset $\mathcal J$ of $\{1,2,\dots,n\}$ such that $I^\diamond_i=0$, for all $i\in\mathcal J$, $I_i^\diamond>0$ for $i\in\{1,2,\dots,n\}\setminus\mathcal J$; and $\beta_i^\diamond>0$ for $i\in\{1,2,\dots,n\}\setminus\mathcal J$, and $\beta_i^\diamond=0$ for $i\in\mathcal J$. WLOG, suppose that $\mathcal J=\{1,2,\dots,s-1\}$ with $s\geq2$. Hence, the endemic relation 
  $\mathbf{I}_h^\diamond=\pi_h(\textrm{diag}(\alpha)-M)^{-1}\mathbf{p}$ and the condition $\bigl\langle
      \Beta   
 \, \mbox{\Large $\mid$} \, (\textrm{diag}(\alpha)-M)^{-1}\mathbf{p}\bigr\rangle=0$ imply that $M$ has the form 
 $$
 M=\begin{pmatrix}
 M_{11} & \textbf{0}_{s-1,n-s+1}\\
 M_{21} & M_{22}
 \end{pmatrix},
 $$
where  $M_{11}\in\mathcal M_{s-1,s-1},\;\; M_{21}\in\mathcal M_{n-s+1,s-1},\;\; M_{22}\in\mathcal M_{n-s+1,n-s+1}$. Similarly, $p_i=0$ for all $i\in\mathcal J$ and 
$\mathbf{I}_h^\diamond=(0,\dots,0, I_s^\diamond,\dots,I_n^\diamond )$ where $I_i^\diamond>0$ for $s\leq i\leq n$. \\

 Let $\mathbf{c}=(\mathbf{c}_1,\mathbf{c}_2)^T$ where $\mathbf{c}_1=(c_1,\dots,c_{s-1})^T$ and $\mathbf{c}_2=(c_s,\dots,c_{n})^T$. The vector $\mathbf{c}_2$  is the solution of $\tilde B\mathbf{c}_2=0$ where
 $$
 \tilde B=\begin{pmatrix}
 -\tilde{b}_{s,s} & m_{s+1,s}I_s^\diamond &  m_{s+2,s}I_{s}^\diamond & \dots & \dots & m_{n,s}I_{s}^\diamond\\
 m_{s,s+1}I_{s+1}^\diamond   & -\tilde{b}_{s+1,s+1} &  m_{s+2,s+1}I_{s+1}^\diamond & \dots &\dots & m_{n,s+1}I_{s+1}^\diamond\\
  m_{s,s+2}I_{s+2}^\diamond &   m_{s+1,s+2}I_{s+2}^\diamond & -\tilde{b}_{33} & \dots& \dots&    m_{n,s+2}I_{s+2}^\diamond\\
 \vdots& \vdots & \vdots& \ddots & \ddots& \vdots \\
m_{s,n}I_{n}^\diamond & m_{s+1,n}I_{n}^\diamond & m_{s+2,n}I_{n}^\diamond  & \hdots&\hdots &-\tilde{b}_{nn}
\end{pmatrix},
 $$
 where $\tilde{b}_{kk}=\sum_{j=s}^nm_{kj}I_j^\diamond$ for $s\leq k\leq n$.
 Since $M_{22}$ is irreducible and $I^\diamond_i>0$ for all $s\leq i\leq n$, the matrix $\tilde B$ is irreducible. Moreover, $\tilde B$ is the Laplacian matrix of the graph interconnecting the stages $I_i$ for $s\leq i\leq n$. Hence, as previously stated, Kirchhoff's matrix tree theorem affirms that the solution of $\tilde B\textbf{c}_2=0$ is such that $c_i=-C_{ii}\gg0$, where $C_{ii}$ is the  cofactor of $i^{th}$ diagonal element of $\tilde B$. Hence $\textbf{c}_2\gg0$.

Let $\mathbf{I}_h^\diamond=(\mathbf{I}_1^\diamond,\mathbf{I}_2^\diamond)^T$ and consider the Lyapunov function candidate $ \mathcal V= \mathcal V_h+ \mathcal V_v$, where
 $$ \mathcal V_h=c_1 \frac{\nu_h}{\alpha_h} E_h+\bigl\langle\textbf{c}_1   \, \mbox{\Large $\mid$} \, \mathbf{I}_1  \bigr\rangle
  +\sum_{i=s}^nc_{i}\int_{I_{i}^\diamond}^{I_{i}}\left(1-\frac{I_{i}^\diamond}{x}\right)dx,\quad\textrm{and}\quad  \mathcal V_v=c_vE_v+c_v\frac{\alpha_v}{\nu_v}I_v,
  $$
  where $c_v=c_1\frac{\nu_h}{\alpha_h}\frac{a\beta_{vh}\Lambda_h}{\mu_hN_h}\frac{\nu_v}{(\mu_v+\delta_v)\alpha_v}$, $\textbf{c}_1$ is positive vector to be determined later. 
  The derivative of $\mathcal V$ along the trajectories of System (\ref{ModelAD2Compact}) is:
    \begin{eqnarray}\label{LyapBoundary}
  \dot {\mathcal V_h} &=&c_1 \frac{\nu_h}{\alpha_h}\dot{E}_h+\bigl\langle\textbf{c}_1   \, \mbox{\Large $\mid$} \, \dot{\mathbf{I}}_1^\diamond  \bigr\rangle
  +\sum_{i=s}^nc_i\left(1-\frac{I_i^\diamond}{I_i}\right)\dot I_i+c_v\dot E_v+c_v\frac{\alpha_v}{\nu_v}\dot I_v\nonumber\\
  &=&c_1 a\,\beta_{vh}\frac{\nu_h}{\alpha_h}\dfrac{S_{h}I_{v}}{N_{h}}+\bigl\langle\textbf{c}_1   \, \mbox{\Large $\mid$} \, (-\textrm{diag}(\tilde\alpha)+M_{11})\mathbf{I}_1  \bigr\rangle+\sum_{i=s}^nc_i\left(1-\frac{I_i^\diamond}{I_i}\right)\dot I_i,
  \end{eqnarray}
  where $\tilde\alpha=(\alpha_1,\dots,\alpha_{s-1})$.
  Moreover, as in the proof of \Cref{GASEEAD}, using the fact, for  that, for $1\leq i\leq n$, the $c_i$ are the components of the solution of $\tilde B\textbf{c}_2=0$ and 
  $$
   \sum_{i=s}^nc_i\sum_{j=s}^nm_{ij}I_j^\diamond\left(1 - \frac{I_i^\diamond}{I_i}\frac{I_j}{I_j^\diamond}\right)= 
   \sum_{i=s}^nc_i\sum_{j=s}^nm_{ij}I_j^\diamond\left(1 - \frac{I_i^\diamond}{I_i}\frac{I_j}{I_j^\diamond}+ \ln\frac{I_i^\diamond}{I_i}\frac{I_j}{I_j^\diamond}\right),
  $$
    it could be shown that Equation (\ref{LyapBoundary}) implies that 
    
       \begin{eqnarray}\label{ciM222}\sum_{i=s}^nc_i\left(1-\frac{I_i^\ast}{I_i}\right)\dot I_i
           &=&\sum_{i=s}^nc_ip_i\pi_h\left(2-\frac{I_i^\diamond}{I_i} -\frac{I_i}{I_i^\diamond}\right)+\sum_{i=s}^nc_i\sum_{j=s}^nm_{ij}I_j^\diamond\left(1 - \frac{I_i^\diamond}{I_i}\frac{I_j}{I_j^\diamond}+\ln\frac{I_i^\diamond}{I_i}\frac{I_j}{I_j^\diamond}\right)\nonumber\\
           & &- \sum_{i=s}^{n}c_i \frac{I_i^\diamond}{I_i}\sum_{i=1}^{s-1}m_{ij}I_j  
            +\sum_{i=s}^nc_i \sum_{j=1}^{s-1}m_{ij}I_j   \nonumber\\
            &:=&\sum_{i=s}^nc_ip_i\pi_h\left(2-\frac{I_i^\ast}{I_i} -\frac{I_i}{I_i^*}\right)+\sum_{i=s}^nc_i\sum_{j=s}^nm_{ij}I_j^\diamond\left(1 - \frac{I_i^\diamond}{I_i}\frac{I_j}{I_j^\diamond}+\ln\frac{I_i^\diamond}{I_i}\frac{I_j}{I_j^\diamond}\right)\nonumber\\
           & &- \sum_{i=s}^{n}c_i \frac{I_i^\diamond}{I_i}\sum_{i=1}^{s-1}m_{ij}I_j +\textbf{c}_2^TM_{21}\mathbf{I}_1.
    \end{eqnarray}
  We choose $\mathbf{c}_1$ to be the solution of $(-\textrm{diag}(\tilde\alpha)+M_{11}^T+\bar\Beta \tilde{e}_1^T)\mathbf{c}_1=-M_{21}^T\mathbf{c}_2$, where $\bar\Beta=\frac{a^2\beta_{vh}\Lambda_h}{\mu_hN_h}\frac{\nu_v\nu_h}{(\mu_v+\delta_v)\alpha_v\alpha_h}\frac{N_v}{N_h}\tilde\Beta$, with $\tilde\Beta=(\beta_1,\dots,\beta_{s-1})$ and $\tilde{e}_1$ the fist canonical vector of $\R^{s-1}$. This solution exists and $\mathbf{c}_1\geq0$ since $\mathbf{c}_2\gg0$ and $-(-\textrm{diag}(\tilde\alpha)+M_{11}^T+\bar\Beta \tilde{e}_1^T)^{-1}\geq0$ as $-\textrm{diag}(\tilde\alpha)+M_{11}^T+\bar\Beta \tilde{e}_1^T$ is a Metzler invertible matrix. \\
  
  Hence, Equations (\ref{LyapBoundary}) and (\ref{ciM222}) leads to
      \begin{eqnarray}\label{LyapBoundary2}
  \dot {\mathcal V_h} 
  &=& c_1 a\,\beta_{vh}\frac{\nu_h}{\alpha_h}\dfrac{S_{h}I_{v}}{N_{h}}+\bigl\langle\textbf{c}_1   \, \mbox{\Large $\mid$} \, (-\textrm{diag}(\tilde\alpha)+M_{11})\mathbf{I}_1  \bigr\rangle +\sum_{i=s}^nc_ip_i\pi_h\left(2-\frac{I_i^\diamond}{I_i} -\frac{I_i}{I_i^\diamond}\right)\nonumber\\
  & &+\sum_{i=s}^nc_i\sum_{j=s}^nm_{ij}I_j^\diamond\left(1 - \frac{I_i^\diamond}{I_i}\frac{I_j}{I_j^\diamond}+\ln\frac{I_i^\diamond}{I_i}\frac{I_j}{I_j^\diamond}\right) 
  - \sum_{i=s}^{n}c_i \frac{I_i^\diamond}{I_i}\sum_{i=1}^{s-1}m_{ij}I_j +\textbf{c}_2^TM_{21}\mathbf{I}_1
  \end{eqnarray}
    However,
    \begin{eqnarray}\bigl\langle\textbf{c}_1   \, \mbox{\Large $\mid$} \, (-\textrm{diag}(\tilde\alpha)+M_{11})\mathbf{I}_1  \bigr\rangle+\textbf{c}_2^TM_{21}\mathbf{I}_1&=&\bigl\langle(-\textrm{diag}(\tilde\alpha)+M_{11}^T)\textbf{c}_1 +M_{21}^T \textbf{c}_2 \, \mbox{\Large $\mid$} \, \mathbf{I}_1  \bigr\rangle\nonumber\\
    &=&-\bigl\langle\bar\Beta \tilde{e}_1^T\textbf{c}_1 \, \mbox{\Large $\mid$} \, \mathbf{I}_1  \bigr\rangle\nonumber\\
    &:=&-c_1\frac{a^2\beta_{vh}\Lambda_h}{\mu_hN_h}\frac{\nu_v\nu_h}{(\mu_v+\delta_v)\alpha_v\alpha_h}\frac{N_v}{N_h}\bigl\langle\tilde\Beta \, \mbox{\Large $\mid$} \, \mathbf{I}_1  \bigr\rangle.\nonumber
    \end{eqnarray}
 Hence, Equation (\ref{LyapBoundary2}) leads to   
        \begin{eqnarray}\label{LyapBoundary3}
  \dot {\mathcal V_h} 
  &=& c_1 a\,\beta_{vh}\frac{\nu_h}{\alpha_h}\dfrac{S_{h}I_{v}}{N_{h}}+\sum_{i=s}^nc_ip_i\pi_h\left(2-\frac{I_i^\diamond}{I_i} -\frac{I_i}{I_i^\diamond}\right)   -c_1\frac{a^2\beta_{vh}\Lambda_h}{\mu_hN_h}\frac{\nu_v\nu_h}{(\mu_v+\delta_v)\alpha_v\alpha_h}\frac{N_v}{N_h}\bigl\langle\tilde\Beta \, \mbox{\Large $\mid$} \, \mathbf{I}_1  \bigr\rangle\nonumber\\
  & &+\sum_{i=s}^nc_i\sum_{j=s}^nm_{ij}I_j^\diamond\left(1 - \frac{I_i^\diamond}{I_i}\frac{I_j}{I_j^\diamond}+\ln\frac{I_i^\diamond}{I_i}\frac{I_j}{I_j^\diamond}\right) 
  - \sum_{i=s}^{n}c_i \frac{I_i^\diamond}{I_i}\sum_{i=1}^{s-1}m_{ij}I_j.  
  \end{eqnarray}
We can check that derivative of $\mathcal V_v$ along the trajectories of (\ref{ModelAD2Compact})
      \begin{eqnarray}\label{LyapVectBoundary}
  \dot {\mathcal V_v} &=&c_v\dot E_v+c_v\frac{\alpha_v}{\nu_v}\dot I_v\nonumber\\
 &=&c_v\left(a\dfrac{ S_{v}}{N_h}\bigl\langle
      \Beta    
 \, \mbox{\Large $\mid$} \, \mathbf{I}_h   
\bigr\rangle-\alpha_vE_{v}\right)
+c_v\frac{\alpha_v}{\nu_v}\left(\nu_vE_v-(\mu_v+\delta_v)I_{v}\right)\nonumber\\
&=&c_1a^2\frac{\beta_{vh}\Lambda_h}{\mu_hN_h}\frac{\nu_v\nu_h}{(\mu_v+\delta_v)\alpha_v\alpha_h} \dfrac{ S_{v}}{N_h}\bigl\langle
      \Beta    
 \, \mbox{\Large $\mid$} \, \mathbf{I}_h   
\bigr\rangle  
-c_1\frac{\nu_h}{\alpha_h}\frac{a\beta_{vh}\Lambda_h}{\mu_hN_h}I_{v} \nonumber\\
&=&c_1a^2\frac{\beta_{vh}\Lambda_h}{\mu_hN_h}\frac{\nu_v\nu_h}{(\mu_v+\delta_v)\alpha_v\alpha_h} \dfrac{ S_{v}}{N_h}\bigl\langle
      \tilde\Beta    
 \, \mbox{\Large $\mid$} \, \mathbf{I}_1  
\bigr\rangle  
-c_1\frac{\nu_h}{\alpha_h}\frac{a\beta_{vh}\Lambda_h}{\mu_hN_h}I_{v},
  \end{eqnarray}
  since $\bigl\langle
      \Beta    
 \, \mbox{\Large $\mid$} \, \mathbf{I}_h   
\bigr\rangle =\bigl\langle
      \tilde\Beta    
 \, \mbox{\Large $\mid$} \, \mathbf{I}_1   
\bigr\rangle$.
Finally, the derivative of $\mathcal V=\mathcal V_h+\mathcal V_v$ along the trajectories of (\ref{ModelAD2Compact}) is obtained by combining Equation (\ref{LyapBoundary3}) and Equation (\ref{LyapVectBoundary}) as follows:
   \begin{eqnarray}\label{CombinedBoundary}
   \dot {\mathcal V}
&=&c_1 a\,\beta_{vh}\frac{\nu_h}{\alpha_h}\dfrac{1}{N_{h}}\left(S_{h}-\frac{\Lambda_h}{\mu_h}  \right)I_{v}   -c_1\frac{a^2\beta_{vh}\Lambda_h}{\mu_hN_h}\frac{\nu_v\nu_h}{(\mu_v+\delta_v)\alpha_v\alpha_h}\frac{1}{N_h}\left(S_v-N_v  \right)\bigl\langle\tilde\Beta \, \mbox{\Large $\mid$} \, \mathbf{I}_1  \bigr\rangle\nonumber\\
  & &+\sum_{i=s}^nc_ip_i\pi_h\left(2-\frac{I_i^\diamond}{I_i} -\frac{I_i}{I_i^\diamond}\right)+\sum_{i=s}^nc_i\sum_{j=s}^nm_{ij}I_j^\diamond\left(1 - \frac{I_i^\diamond}{I_i}\frac{I_j}{I_j^\diamond}+\ln\frac{I_i^\diamond}{I_i}\frac{I_j}{I_j^\diamond}\right) 
  - \sum_{i=s}^{n}c_i \frac{I_i^\diamond}{I_i}\sum_{i=1}^{s-1}m_{ij}I_j.\nonumber
\end{eqnarray}
  Moreover, using the equation of $\dot S_h$ and $\dot S_v$ in Model (\ref{ModelAD2Compact}), it is straightforward that $S_h\leq\frac{\Lambda_h}{\mu_h}$ and $S_v\leq N_v:=\frac{\pi_v}{\mu_v+\delta_v}$, where $\displaystyle\Lambda_h=\pi_h\left( 1-p-\sum_{i=1}^{n+1}p_i \right)$. Hence $\dot{\mathcal V}\leq0$. Therefore, by Lyapunov's theorem this proves the stability of the weakly endemic equilibrium $({S}_h^\diamond, 0,\mathbf{I}_h^\diamond,{S}_v^0,0,0)$. Furthermore, $\dot{\mathcal V}$ is the sum of five nonpositive terms, of which two are definite-negative. Hence, it is straightforward that the largest invariant on which $\dot{\mathcal V}=0$ is reduced to $({S}_h^\diamond, 0,\mathbf{I}_h^\diamond,{S}_v^0,0,0)$. Thus, by LaSalle's principle,  $({S}_h^\diamond, 0,\mathbf{I}_h^\diamond,{S}_v^0,0,0)$ is asymptotically stable. This completes the proof of the global asymptotic stability of the weakly endemic equilibrium $({S}_h^\diamond, 0,\mathbf{I}_h^\diamond,{S}_v^0,0,0)$.
\end{proof}

Per \Cref{ExistenceEEAD}, \Cref{iv}, a necessary condition to break the host-vector transmission, that is, to maintain the vectors disease-free, is $p=0$ and $\bigl\langle
      \Beta   
 \, \mbox{\Large $\mid$} \, (\textrm{diag}(\alpha)-M)^{-1}\mathbf{p}\bigr\rangle=0$. The later quantity has an epidemiological interpretation. Indeed, it means that: a.) there is an influx of infected individuals only to a subset of indices and that the hosts in these stages are unable to infect the vectors and b.) the infectious hosts at these stages do not ``ameliorate" their infectiosity to stages in the complement of the subset in which they belong. That is, $\delta_{ij}=0$ for all $i\in A$ and $j\in\{1,2,\dots,n\}\setminus A$, with $p_j>0$ for all $j\in A$ and $p_j=0$ otherwise. In this case, the threshold $\mathcal N_0^2(p,\mathbf{p},p_{n+1})$ determine whether or not the vector populations become disease-free. If $\mathcal N_0^2(p,\mathbf{p},p_{n+1})<1$, the disease dies out in the vector population and it thus, the infectious hosts are contained only into the classes in which they are replenished.  This threshold captures the capacity of hosts in stage $A$ to maintain the disease in the vector population. Indeed, we can show that:
\begin{eqnarray*}
 \mathcal N_0^2(p,\mathbf{p},p_{n+1})&=&\frac{a^2\beta_{vh}\nu_v\nu_h}{\alpha_v(\mu_v+\delta_v)\alpha_h\mu_h}\frac{N_v}{N_h}\frac{\pi_h\left( 1-p-\sum_{i=1}^{n+1}p_i \right)}{N_h}\bigl\langle
      \Beta   
 \, \mbox{\Large $\mid$} \, (\textrm{diag}(\alpha)-M)^{-1} e_1  
\bigr\rangle\\
&:=&\frac{a^2\beta_{vh}\nu_v\nu_h}{\alpha_v(\mu_v+\delta_v)\alpha_h\mu_h}\frac{N_v}{N_h}\frac{\pi_h\left( 1-p-\sum_{i=1}^{n+1}p_i \right)}{N_h}\bigl\langle
      \tilde\Beta   
 \, \mbox{\Large $\mid$} \, (\textrm{diag}(\tilde\alpha)-M_{11})^{-1} e_1  
\bigr\rangle
\end{eqnarray*}
\subsection{Sharp threshold property}\label{SharpThreshold}
In this subsection, we investigate the dynamics of Model (\ref{ModelAD2Compact}) when $p=p_1=\dots,p_n=0$. In this case, we obtain the model
\begin{equation} \label{ModelAD2CompactSharp}
\left\{\begin{array}{llll}
\displaystyle\dot S_{h}=\pi_h-a\,\beta_{vh}\,S_{h}\,\dfrac{I_{v}}{N_{h}}-\mu_h S_h\\
\displaystyle\dot E_{h}= a\,\beta_{vh}\,S_{h}\,\dfrac{I_{v}}{N_{h}}-(\mu_h+\nu_h+\eta)E_h\\
\displaystyle\dot{\mathbf{I}}_h=\nu_hE_he_1-(\textrm{diag}(\alpha)-M)\mathbf{I}_h\\
\displaystyle\dot S_{v}=\Lambda_v
 -a\dfrac{ S_{v}}{N_h}\bigl\langle
      \Beta^T   
 \, \mbox{\Large $\mid$} \, \mathbf{I}_h   
\bigr\rangle
-(\mu_v+\delta_v)S_{v}\\
\displaystyle\dot E_{v}=a\dfrac{ S_{v}}{N_h}\bigl\langle
      \Beta^T   
 \, \mbox{\Large $\mid$} \, \mathbf{I}_h   
\bigr\rangle-(\mu_v+\nu_v+\delta_v)E_{v}\\
\displaystyle\dot I_{v}=\nu_vE_v-(\mu_v+\delta_v)I_{v}
\end{array}\right.
\end{equation}
For the same reason evoked in \Cref{sec:SPHostFlux}, the solutions of System (\ref{ModelAD2CompactSharp}) stay positive and bounded. Unlike in Model (\ref{ModelAD2Compact}),  the Model (\ref{ModelAD2CompactSharp}) has a disease free equilibrium (DFE), and is given by $(S_h^0,0,0,S_v^0,0,0)$ with  $S_h^0=\frac{\pi_h}{\mu_h}$ and $S_v^0=\frac{\pi_v}{\mu_v+\delta_v}$.

The basic reproduction number $\mathcal R_0^2$ is derived using the next generation method. An explicit expression of it is given by
\begin{eqnarray*}
\mathcal R_0^2&=&\frac{a^2\beta_{vh}\nu_h\nu_vN_v}{(\mu_h+\nu_h+\eta)(\nu_v+\mu_v+\delta_v)(\nu_v+\delta_v)N_h}\Beta^T(\textrm{diag}(\alpha)-M)^{-1}e_1\\
&:=&\mathcal N_0^2(0,\mathbf{0},0).
\end{eqnarray*}
 Note that since the matrix is $M$ is Metzler (off-diagonal elements are non-negative) and invertible, we have $-M^{-1}\geq0$. Thus, $\mathcal R_0^2\geq0$. The following theorem gives the complete asymptotic behavior of Model (\ref{ModelAD2CompactSharp}).
\begin{theorem}\label{theo:SharpThreshold}\hfill
\begin{enumerate}
\item If $\mathcal R_0^2\leq1$, the DFE is globally asymptotically stable.
\item If $\mathcal R_0^2>1$, the DFE is unstable and a unique endemic equilibrium exists and is GAS.
\end{enumerate}
\end{theorem}
The proof of the first part of \Cref{theo:SharpThreshold} follows using, for example, a left-eigenvector argument. We omit the details. The second part is particular case of \Cref{GASEEAD}. This result is new in itself.
 \section{Illustrations and Simulations}\label{IllustrationsNSimulations}
 In this section, we provide illustrations to highlight the effects of influx of immigrants and the transfer matrix on the disease dynamics and  provide some numerical simulations to showcase the results of \Cref{sec:GSVBAP}. To do so, we consider the case $n=4$. That is, there are four infectious stages in the host's infectivity. Unless otherwise stated, we consider the following baseline parameters:
 $$
 \pi_h=1000, a=0.7, \beta_{vh}=0.3,
\mu_h=\frac{1}{75\times365}\;\textrm{days}^{-1},\;
\nu_h=\frac{1}{15}\;\textrm{days}^{-1},\; 
\gamma_{12}=\frac{1}{8}\;\textrm{days}^{-1},\;  
   $$
$$
 \gamma_{23}=\gamma_{34}=\frac{1}{6}\;\textrm{days}^{-1},\;
  \frac{1}{\eta}=0\;\textrm{days}^{-1},\; 
 \eta_1=\frac{1}{50}\;\textrm{days}^{-1},\; 
\eta_2=\eta_3=\frac{1}{30}\;\textrm{days}^{-1},\;
\eta_4=\frac{1}{40}\;\textrm{days}^{-1},\; 
$$
$$\pi_v=10000, 
\mu_v=\frac{1}{15}\;\textrm{days}^{-1},\;  
\nu_v=\frac{1}{4}\;\textrm{days}^{-1},\;  
\delta_v=\frac{1}{20}\;\textrm{days}^{-1}.
 $$
 It is worthwhile noting that, although reasonable, these values do not necessarily match any particular arbovirus diseases. We have chosen them to encompass results of \Cref{sec:GSVBAP}. The transfer matrix $M$ and the vector proportions of influx of infected $\mathbf{p}$ are given  by
 $$
 M=\begin{pmatrix}
 0 & \delta_{21} &\delta_{31} &\delta_{41}\\
  \gamma_{12} &0 &\delta_{32} &\delta_{42}\\
   \gamma_{13} & \gamma_{23} &0 &\delta_{43}\\
    \gamma_{14} & \gamma_{24} &\gamma_{34} &0  
 \end{pmatrix},\quad
  \mathbf{p}=\begin{pmatrix}
p_1\\
 p_2\\
p_3\\
p_4 
 \end{pmatrix}.
 $$
 We vary the parameter $p$, the vector $\mathbf{p}$ and the matrix $M$ to investigate their impacts on the disease dynamics.
 
 \begin{figure}[ht]
\centering
 \subfigure[Dynamics of infectious hosts $I_i$, for $i=1,\dots,4$ when $\Beta=\mathbf{0}_{\R^n}$.]{
   \includegraphics[scale =.36]{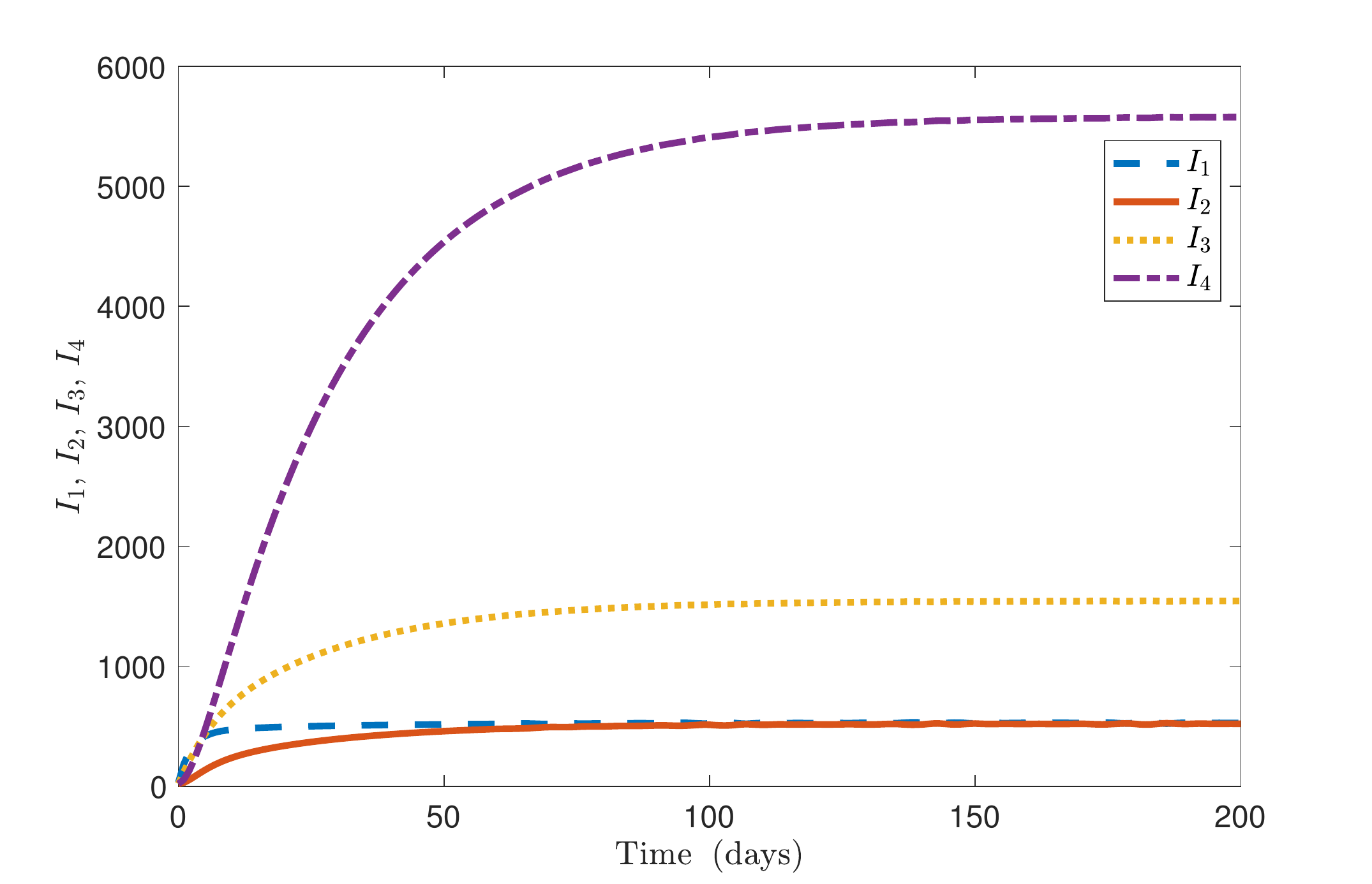}
\label{WealkyEndemicVectorsHost}}
\hspace{1mm}
 \subfigure[Dynamics of infected vectors when $\Beta=\mathbf{0}_{\R^n}$.]{
   \includegraphics[scale =.36]{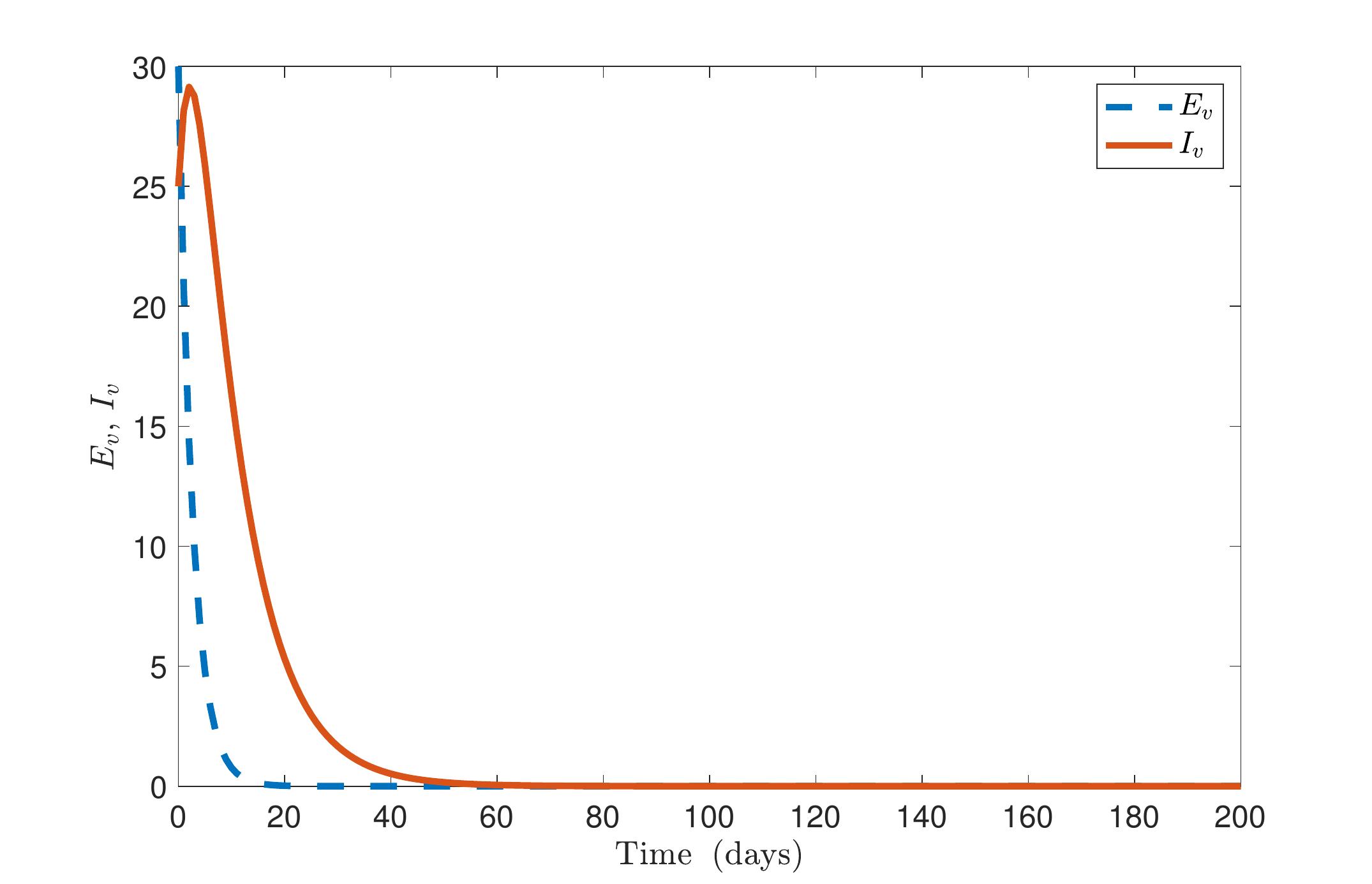}
\label{WealkyEndemicVector}}\\
 \subfigure[Dynamics of infected hosts when $\Beta\neq\mathbf{0}_{\R^4}$ and $p=0.01\neq0$]{
   \includegraphics[scale =.35]{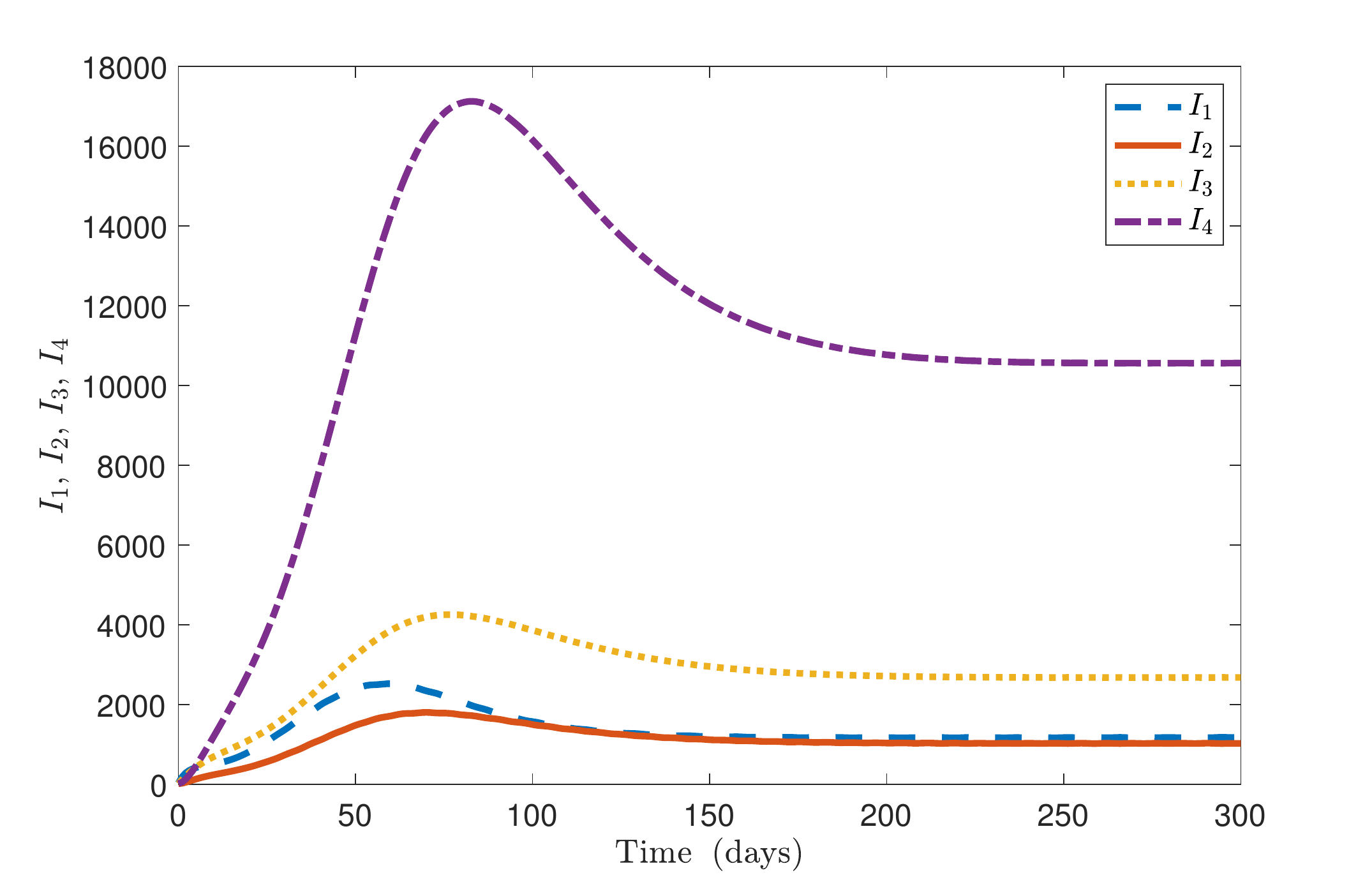}
\label{SEEBetaNpNonzeroHost}}
\hspace{1mm}
 \subfigure[Dynamics of infected vectors when $\Beta\neq\mathbf{0}_{\R^n}$ and $p=0.01\neq0$.]{
   \includegraphics[scale =.35]{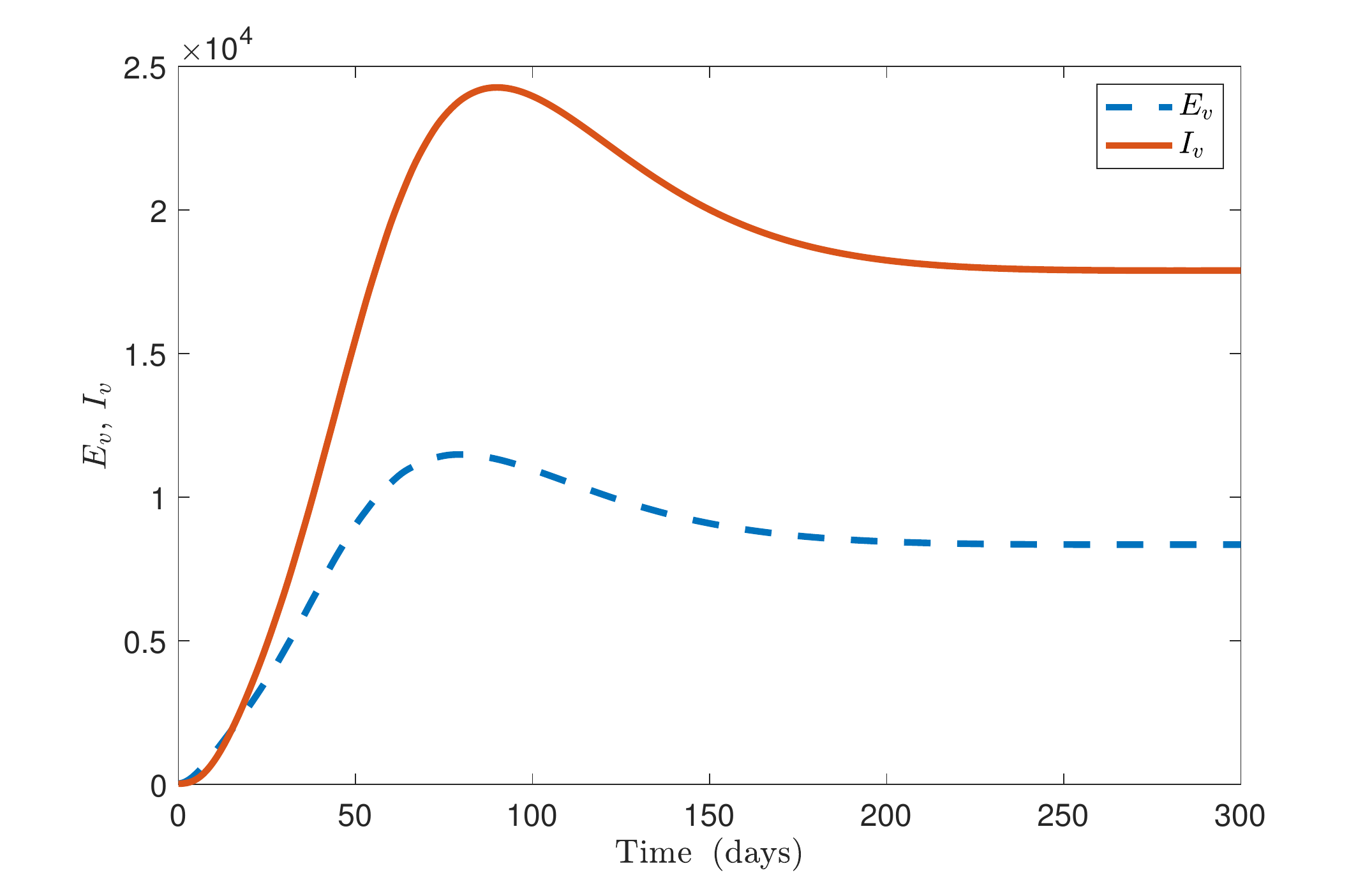}
\label{SEEBetaNpNonzeroVectors}}
\caption{Effects of host-vector transmission on the dynamics of Model (\ref{ModelAD2Compact}) with $n=4$. The proportions of infectious influx are $p_1=0.2$, $p_3=0.1$, $p_4=0$ and $p_{5}=0.3$. The transfer matrix $M$ is such as $\gamma_{13}=\gamma_{24}=0.1$, $\gamma_{14}=0.2$, $\delta_{21}=0.01$, $\delta_{31}=0.02$,  $\delta_{41}=0.001$, $\delta_{32}=0.03$, $\delta_{42}=0.01$ and $\delta_{43}=0.03$.
 } \label{fig:EffectsBeta}
\end{figure}

 \Cref{WealkyEndemicVectorsHost} and \Cref{WealkyEndemicVector} depict the dynamics of the Model (\ref{ModelAD2Compact}) when there is no transmission form hosts to vectors. That is, whenever $\Beta=\mathbf{0}_{\R^n}$. In this case, with $\mathbf{p}\neq\mathbf{0}_{\R^n}$, the infected hosts reach an endemic level ( \Cref{WealkyEndemicVectorsHost}) while the disease dies out in the vector population (\Cref{WealkyEndemicVector}). This is in accordance is the prediction of \Cref{ExistenceEEAD}, \Cref{i}, where the weakly endemic equilibrium is GAS (\Cref{GASEEAD}).
 
For $\Beta=(0.2,0,0,0.5)^T\neq\mathbf{0}_{\R^4}$ and $p=0.01$, the trajectories converge to a strongly endemic equilibrium (\Cref{SEEBetaNpNonzeroHost} and \Cref{SEEBetaNpNonzeroVectors}) as the hypotheses of \Cref{ExistenceEEAD}, \Cref{ii} are satisfied. \\ 
 
 To illustrate \Cref{ExistenceEEAD}, \Cref{iii}, suppose that $\Beta=(\beta_1,0,0,0)^T$ and $\mathbf{p}=(0,0,p_3,p_4)^T$ where $\beta_1>0$, $p_3>0$ and $p_4>0$. By choosing $\gamma_{14}=\gamma_{24}=\delta_{21}=\delta_{31}=\delta_{42}=0$, we obtain:
 
 \begin{equation}\label{BetaTM}\Beta^T(\textrm{diag}(\alpha)-M)^{-1}\mathbf{p}=\frac{\beta_1\alpha_2\delta_{41}(p_3\gamma_{34}+p_4\alpha_3)}{\det(\textrm{diag}(\alpha)-M))}>0.\end{equation}
 Using this setup, \Cref{iii} of \Cref{ExistenceEEAD} anticipates the existence of an strongly endemic equilibrium. Indeed, \Cref{fig:Case3} represents the dynamics of hosts (\Cref{fig:HostCase3}) and vectors (\Cref{fig:VectorCase3}) in  Model (\ref{ModelAD2Compact}) in this case.

  \begin{figure}[H]
\centering
 \subfigure[Dynamics of infectious hosts $I_i$, for $i=1,\dots,4$.]{
   \includegraphics[scale =.28]{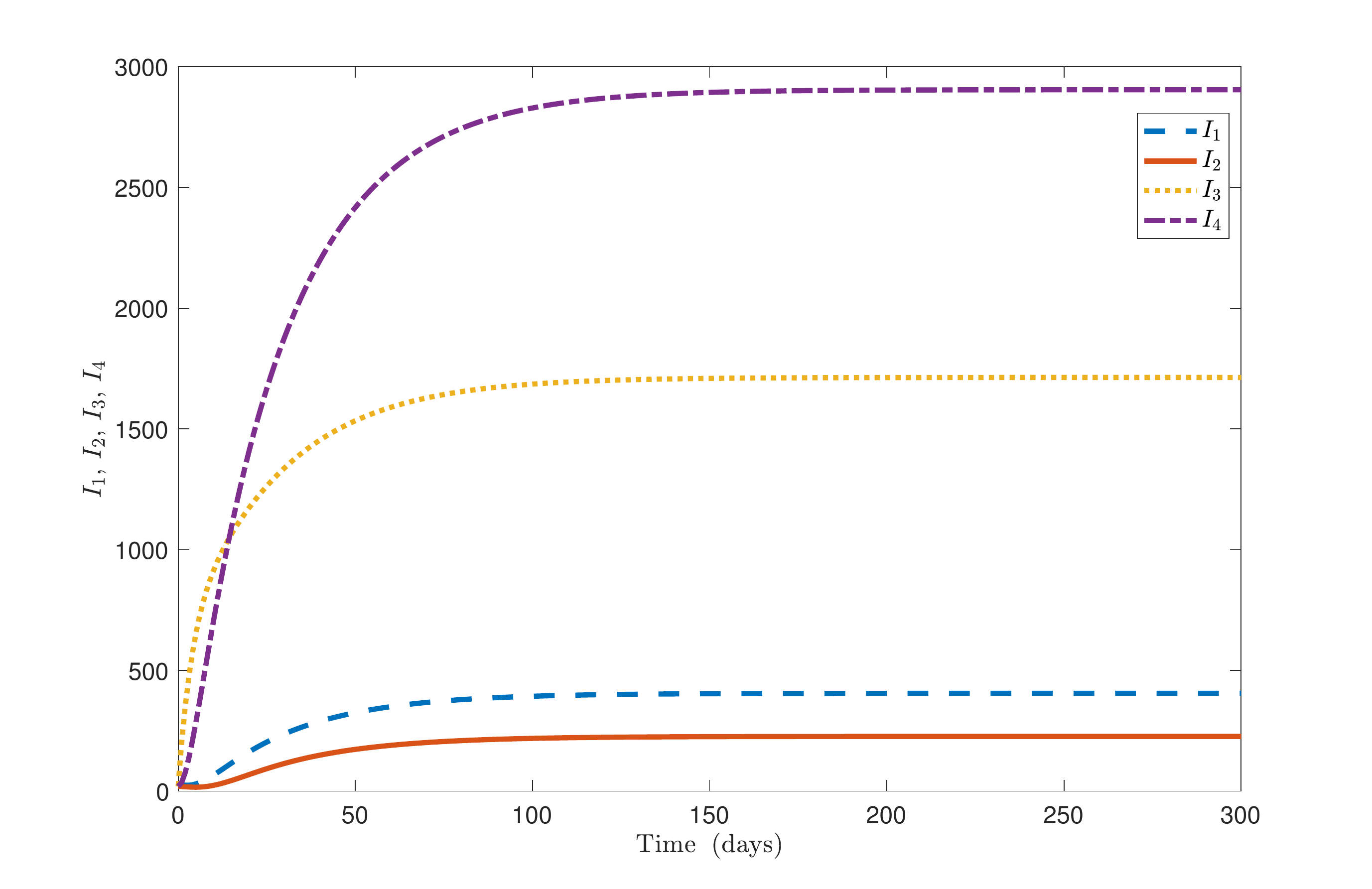}
\label{fig:HostCase3}}
\hspace{1mm}
 \subfigure[Dynamics of infected vectors.]{
   \includegraphics[scale =.28]{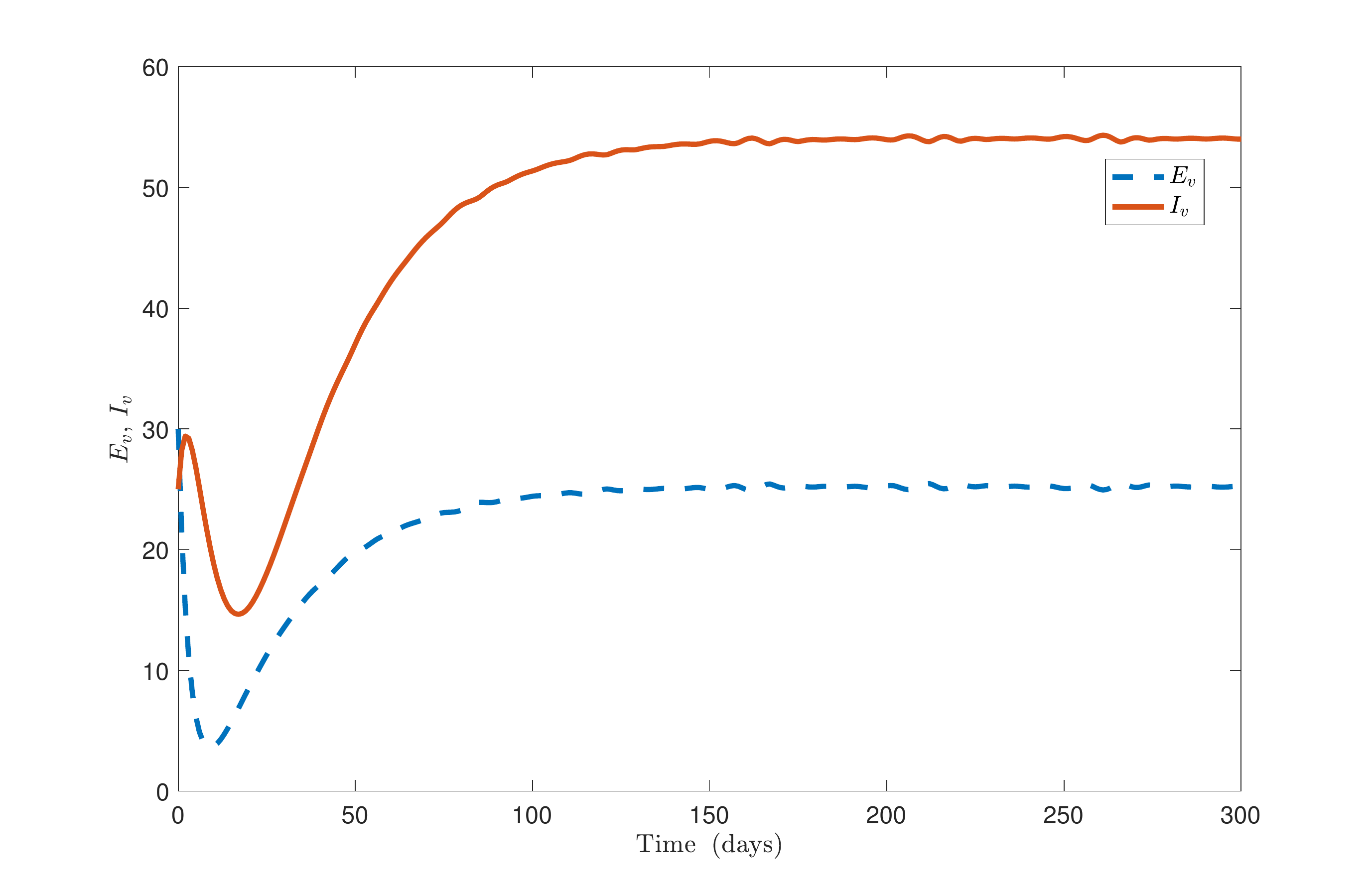}
\label{fig:VectorCase3}}%
\caption{Dynamics of infected hosts and vectors when the hypotheses of \Cref{ExistenceEEAD}, \Cref{iii} are satisfied. The proportions of infectious influx are $p=0$, $\mathbf{p}=(0,0,p_3,p_4)^T=(0,0,0.2,0.0001)^T$ and $p_{5}=0.3$. The transfer matrix $M$ is such as $\gamma_{13}=0.1$, $\gamma_{14}=\gamma_{24}=0$, $\delta_{21}=0.01$, $\delta_{31}=\delta_{32}=\delta_{42}=0$, $\delta_{41}=0.035$,  and $\delta_{43}=0.03$.
 } \label{fig:Case3}
\end{figure}
 
 By choosing $\delta_{41}=0$, \Cref{BetaTM} implies that $\Beta^T(\textrm{diag}(\alpha)-M)^{-1}\mathbf{p}=0$ and thus satisfying the conditions of \Cref{ExistenceEEAD}, \Cref{iv}. And so, a weakly endemic equilibrium $({S}_h^\diamond, 0,\mathbf{I}_h^\diamond,{S}_v^0,0,0)$ or a strongly endemic equilibrium $(\tilde{S}_h, \tilde{E}_h,\tilde{\mathbf{I}}_h,\tilde{S}_v,\tilde{E}_v,\tilde{I}_v)$ exists depending on whether ${\mathcal N}_0^2(p,\mathbf{p},p_{n+1})$ is below or greater than unity, respectively. \Cref{fig:HostsCase4N0Less1} shows that the hosts' infection dies out at stage 1 and 2 while it persists at stage 3 and 4. The disease dies out the vectors' population (\Cref{fig:VectorsCase4N0Less1}). It is worthwhile noting that the disease is maintained at stages 1 and 2, due to the influx of infectious individuals at these stages, without whom, the interaction between  hosts and vectors is not sufficient to sustain the infectious. That is, ${\mathcal N}_0^2(p,\mathbf{p},p_{n+1})\leq1$.  Under the same transfer matrix $M$ and the infectious influx $\mathbf{p}$ configurations, but choosing the entomological parameters $a=0.9$ and $\beta_{vh}=0.9$, we obtain  ${\mathcal N}_0^2(p,\mathbf{p},p_{n+1})=1.6051>1$. This leads to a strongly endemic equilibrium (\Cref{fig:HostsCase4N0Greater1} and \Cref{fig:VectorsCase4N0Greater1}).
  \begin{figure}[ht]
\centering
 \subfigure[Dynamics of infectious hosts $I_i$, for $i=1,\dots,4$ when ${\mathcal N}_0^2=0.0066\leq1$.]{
   \includegraphics[scale =.28]{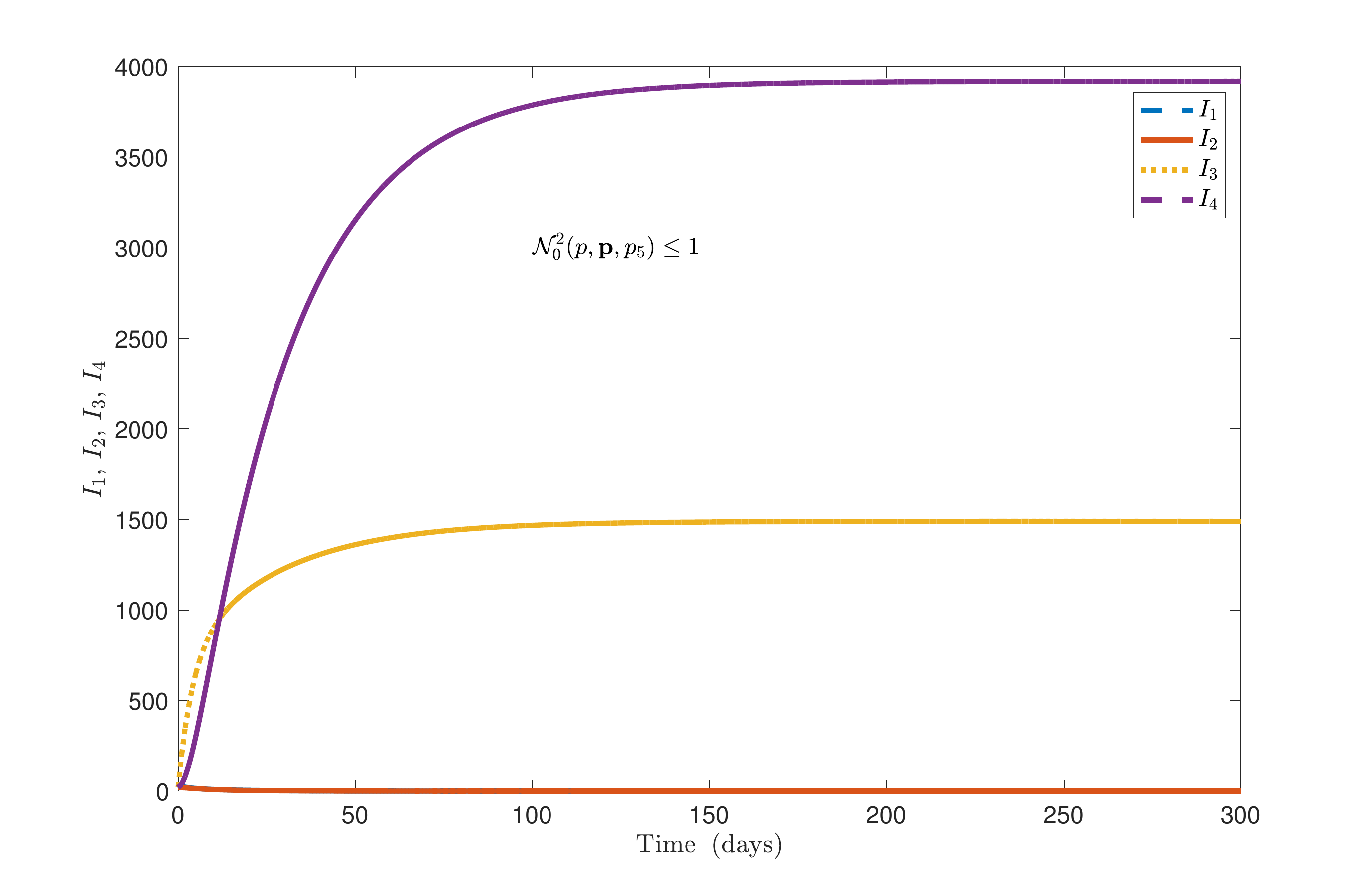}
\label{fig:HostsCase4N0Less1}}
\hspace{1mm}
 \subfigure[Dynamics of infected vectors when ${\mathcal N}_0^2=0.0066\leq1$.]{
   \includegraphics[scale =.28]{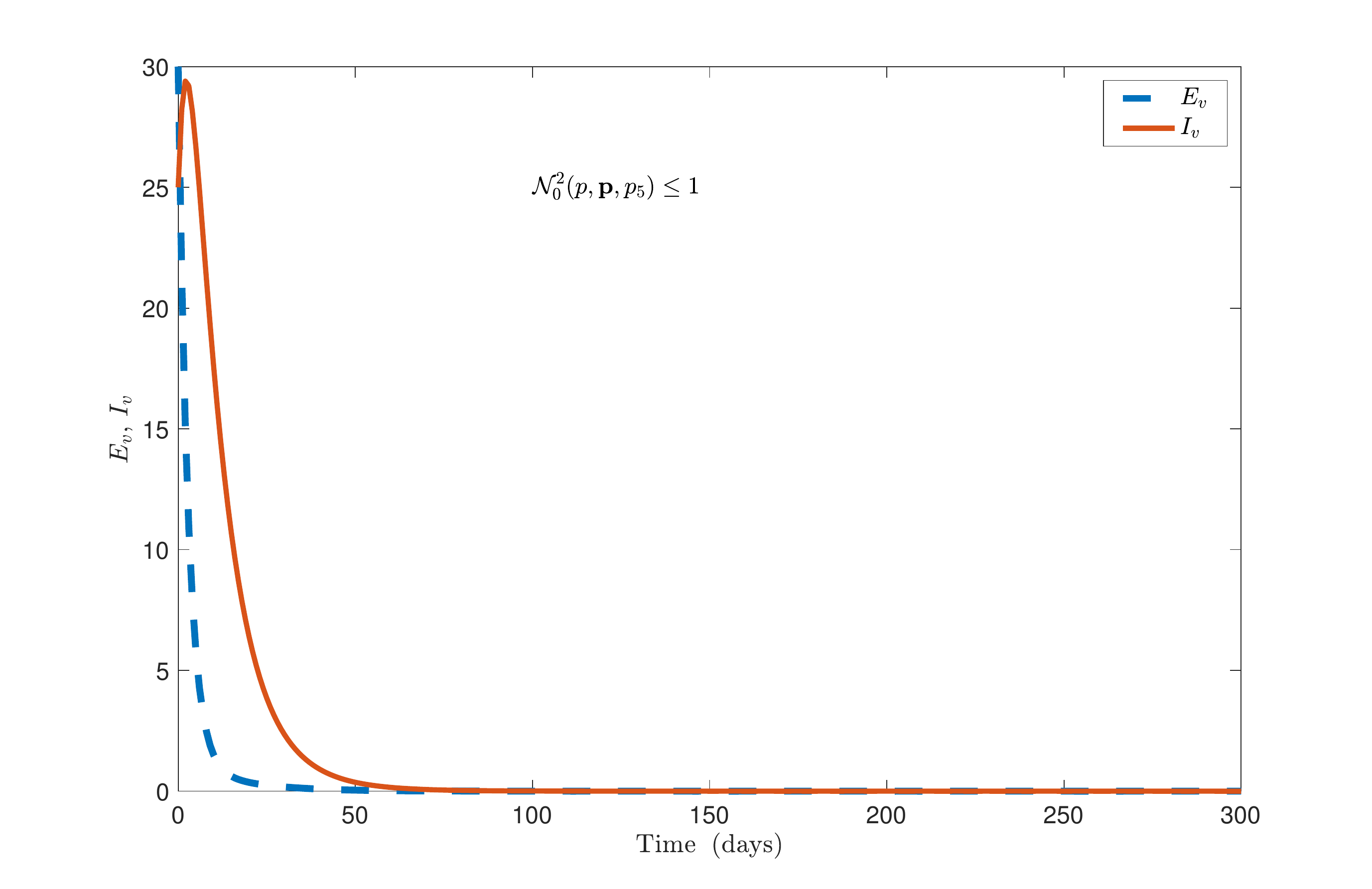}
\label{fig:VectorsCase4N0Less1}}\\
 \subfigure[Dynamics of infectious hosts $I_i$, for $i=1,\dots,4$ when ${\mathcal N}_0^2>1$]{
   \includegraphics[scale =.28]{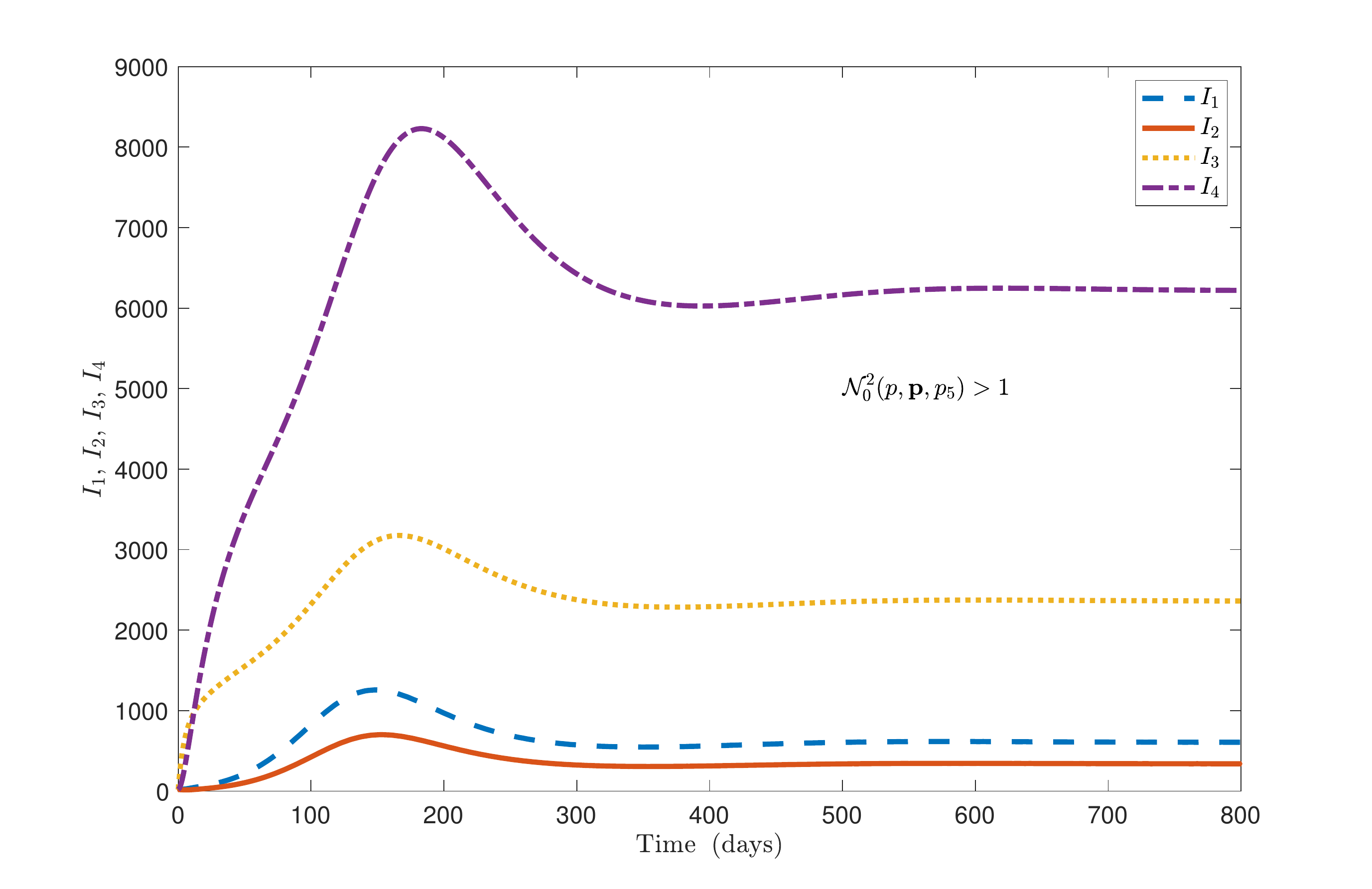}
\label{fig:HostsCase4N0Greater1}}
\hspace{1mm}
 \subfigure[Dynamics of infected vectors when ${\mathcal N}_0^2<1$.]{
   \includegraphics[scale =.28]{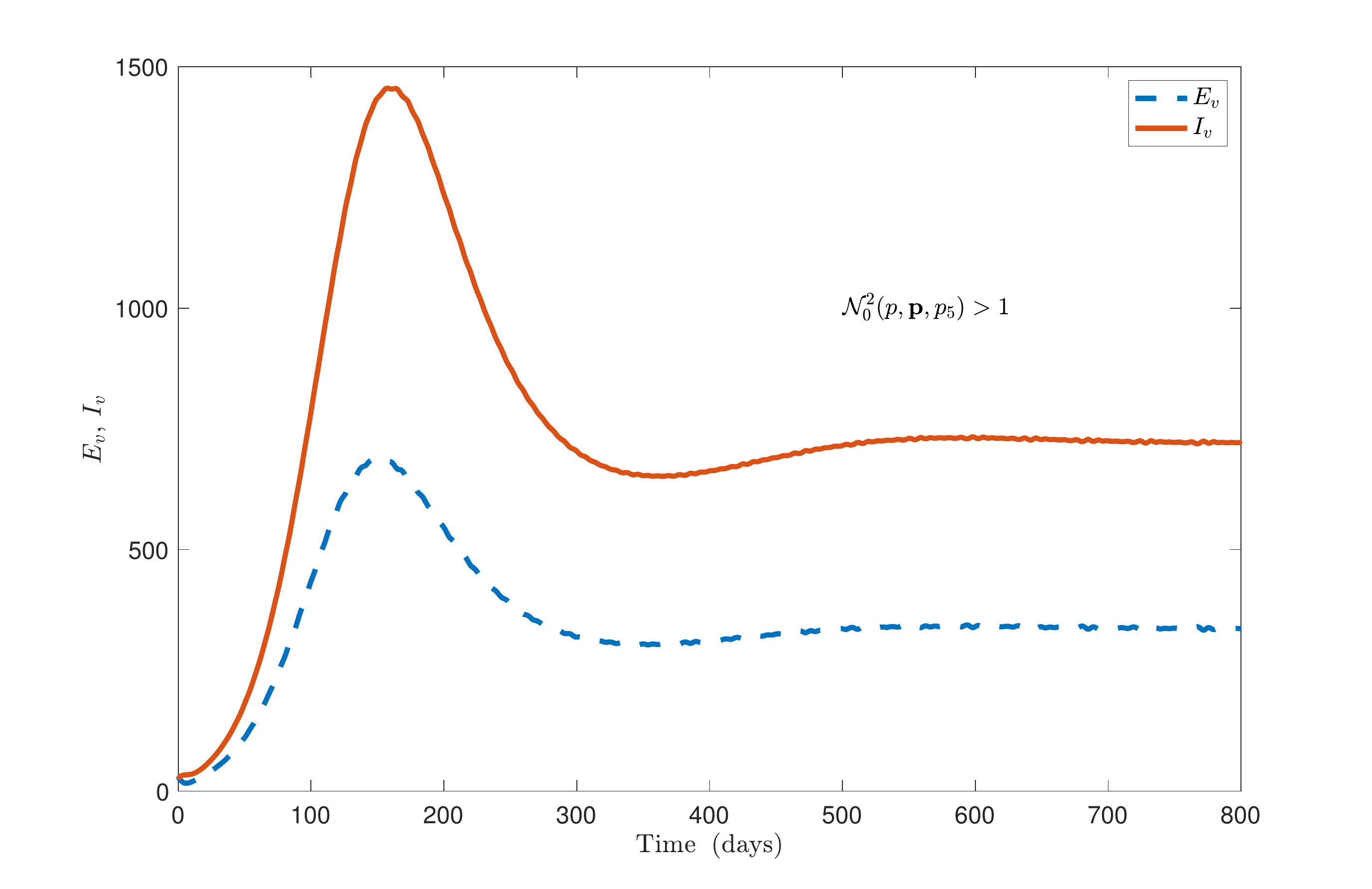}
\label{fig:VectorsCase4N0Greater1}}
\caption{Dynamics of infected hosts and vectors when the hypotheses of \Cref{ExistenceEEAD}, \Cref{iv} are satisfied. The proportions of infectious influx are $p=0$, $\mathbf{p}=(0,0,p_3,p_4)^T=(0,0,0.2,0.0001)^T$ and $p_{5}=0.3$. The transfer matrix $M$ is such as $\gamma_{13}=0.1$, $\gamma_{14}=\gamma_{24}=0$, $\delta_{21}=0.01$, $\delta_{31}=\delta_{32}=\delta_{41}=\delta_{42}=0$,  and $\delta_{43}=0.03$. 
 } \label{fig:Case4N0LessNGreater1}  
\end{figure}

\section{Conclusion}
Modeling the dynamics of vector-borne diseases have often been based on the assumption that the recruitment into the population is completely susceptible, and thereby making it difficult to assess the effects of the infected or infectious individuals who enters the population. However, the recent surge of vector-borne diseases such as Chikunguyna and Dengue in areas previously free from the pre-cited diseases, that also  coincides with an increase of global travel across the world, makes the study of the effects of new arrivals on vector-borne diseases dynamics a necessity. Indeed, the arrival of new individuals from endemic areas, or the return of local residents after a stint in areas where the vector-borne diseases are endemic, could potentially result in infecting the local vector populations and the cycle of host-vector infection could start or accelerate.\\

In this paper, we formulate a general staged-progression and stage-regression vector-borne diseases to capture some key features of their dynamics. Particularly, we investigate the effects of the, often swept under the rug, influx of viremic individuals into the population and vectors' dynamics. We also explore the impacts of treatment and repeated exposure.  Indeed, assuming the infectious individuals in the population are undergoing a treatment program, whereby improving their health status; this could lead an infectious individual to go from stage $i$ to a ``lower" stage $i-k$, where $1\leq k\leq i$.  Similarly, the repeated exposure of infected hosts to infected vectors could lead to more infectious bite. This could lead to increase in infected hosts' level of parasitemia and thus worsening its health status. In this case, the infected host progresses from stage $i$ to an ``upper" stage $i+k$ where $1\leq i\leq k$. And so, we incorporate of these two phenomena of progression and regression on the hosts' dynamics, which happens to have an altering effects on the qualitative dynamics of the model.\\

We derive an staged-progression vector-borne model with $n$ infectious stages. The host-vector dynamics follows and $SEI^nR-SEI$ framework. We assume that a proportion --of the overall recruitment-- $p$ and $p_i$, for $i=1,2,\dots,n$, of latent and infectious of stage $i$, respectively, enter into the population. An infectious host at stage $i$ could improve its status from stage $i$ to $i+j$, with $i\leq j$, at a rate $\delta_{ij}$ or worsen its viremicity form stage $i$ to stage $ i\geq j$, at a rate $\gamma_{ij}$.  We derived all steady states of the general system and provided conditions under which they exists (\Cref{ExistenceEEAD}). Its turns out that the model has multiple equilibria, depending on the connectivity configuration between host's infectious stages and the influx of infectious arrivals. However, all of the equilibria are either strongly endemic (SEE) --for which all of the infected and infectious components are positive--, or weakly endemic (WEE)--for which some of host and vectors' infected or infectious classes are zero. We show the influx of latent individuals into the population guarantees the existence of an SEE, which is globally asymptotically stable (\Cref{GASEEAD}). This case is particularly important for controlling vector-borne diseases as it pertains to public health policies since it is difficult to detect latent even if screening measures were in place.  \\

When there is no influx of latent but the host-vector transmission vector $\Beta$, the vector of influx of infectious $\mathbf{p}$ and transfer rates matrix $M$ are such that $\Beta^T(\textrm{diag}(\alpha)-M)^{-1}\mathbf{p}=0$, that is, whenever the infectious stages with non zero influx do not transmit the infection to the susceptible vector population. We show that if there is an influx in all infectious stages --$\mathbf{p}\gg0$--, then   $\Beta=0$ and the disease dies out in the vector populations (see \Cref{WEEGAS} and the first part of its proof). If $\Beta\neq0$, then there will not be influx of infectious in all stages. Moreover, if  infectious hosts in the stages with influx do not improve their to stages that are capable of transmitting the infection to vectors, that is stages where $\beta_i>0$, then a threshold $\mathcal N_0^2(p,\mathbf{p},p_{n+1})$ arises. The disease will die out in the vector population if $\mathcal N_0^2(p,\mathbf{p},p_{n+1})$ is below unity and persists otherwise.\\

Our results show that when there is no influx of infected and infectious individuals, the considered model becomes a vector-borne disease with $n$ infectious stages that accounts for amelioration from and progression to any stages.  We show that this model has a sharp threshold phenomenon, for which the dynamics is completely determined by the basic reproduction number $\mathcal R_0^2$ (\Cref{theo:SharpThreshold}). It turns out that  $\mathcal R_0^2=\mathcal N_0^2(0,\mathbf{0},0)$, and if $\mathcal R_0^2\leq1$, the disease-free equilibrium exists and is globally asymptotically stable. Moreover, if $\mathcal R_0^2>1$, we show that an endemic equilibrium exists and is globally asymptotically stable.


\begin{thebibliography}{10}

\bibitem{barnett2008role}
{\sc E.~D. Barnett and P.~F. Walker}, {\em Role of immigrants and migrants in
  emerging infectious diseases}, Medical Clinics of North America, 92 (2008),
  pp.~1447--1458.

\bibitem{benedict2007spread}
{\sc M.~Q. Benedict, R.~S. Levine, W.~A. Hawley, and L.~P. Lounibos}, {\em
  Spread of the tiger: global risk of invasion by the mosquito aedes
  albopictus}, Vector-borne and zoonotic Diseases, 7 (2007), pp.~76--85.

\bibitem{BicharaIggidrSmith2017}
{\sc D.~Bichara, A.~Iggidr, and L.~Smith}, {\em Multi-stage vector-borne
  zoonoses models: A global analysis}, Bulletin of Mathematical Biology,
  https://doi.org/10.1007/s11538-018-0435-1 (2018).

\bibitem{bollobas2013modern}
{\sc B.~Bollob{\'a}s}, {\em Modern graph theory}, Springer Science \& Business
  Media, 2013.

\bibitem{BraVdd01}
{\sc F.~Brauer and P.~van~den Driessche}, {\em Models for transmission of
  disease with immigration of infectives}, Math. Biosci., 171 (2001).

\bibitem{CasThi95}
{\sc C.~Castillo-Chavez and H.~R. Thieme}, {\em Asymptotically autonomous
  epidemic models}, in Mathematical Population Dynamics: Analysis of
  Heterogeneity, Volume One: Theory of Epidemics,, O.~Arino, A.~D.E., and
  M.~Kimmel, eds., Wuerz, 1995.

\bibitem{CDCVectorBorneUSA}
{\sc {Centers for Disease Control and Prevention}}, {\em Illnesses from
  mosquito, tick, and flea bites increasing in the us}, Centers for Disease
  Control and Prevention,
  \url{https://www.cdc.gov/media/releases/2018/p0501-vs-vector-borne.html},
  (2018).

\bibitem{cruz2012control}
{\sc G.~Cruz-Pacheco, L.~Esteva, and C.~Vargas}, {\em Control measures for
  chagas disease}, Mathematical biosciences, 237 (2012), pp.~49--60.

\bibitem{Esteva98Dengue}
{\sc L.~Esteva and C.~Vargas}, {\em Analysis of a dengue disease transmission
  model}, Math. Biosci., 150 (1998), pp.~131--151.

\bibitem{gould2010first}
{\sc E.~A. Gould, P.~Gallian, X.~De~Lamballerie, and R.~N. Charrel}, {\em First
  cases of autochthonous dengue fever and chikungunya fever in france: from bad
  dream to reality!}, Clinical microbiology and infection, 16 (2010),
  pp.~1702--1704.

\bibitem{grandadam2011chikungunya}
{\sc M.~Grandadam, V.~Caro, S.~Plumet, J.-M. Thiberge, Y.~Souares, A.-B.
  Failloux, H.~J. Tolou, M.~Budelot, D.~Cosserat, I.~Leparc-Goffart, and
  P.~Despr{\`e}s}, {\em Chikungunya virus, southeastern france}, Emerging
  infectious diseases, 17 (2011), p.~910.

\bibitem{guo2012impacts}
{\sc H.~Guo and M.~Y. Li}, {\em Impacts of migration and immigration on disease
  transmission dynamics in heterogeneous populations}, Discrete Contin. Dyn.
  Syst. Ser. B, 17 (2012), pp.~2413--2430.

\bibitem{guo2012global}
{\sc H.~Guo, M.~Y. Li, and Z.~Shuai}, {\em Global dynamics of a general class
  of multistage models for infectious diseases}, SIAM Journal on applied
  mathematics, 72 (2012), pp.~261--279.

\bibitem{isaacson1989airport}
{\sc M.~Isa{\"a}cson}, {\em Airport malaria: a review}, Bulletin of the World
  Health Organization, 67 (1989), p.~737.

\bibitem{johnson2016modeling}
{\sc T.~L. Johnson, E.~L. Landguth, and E.~F. Stone}, {\em Modeling relapsing
  disease dynamics in a host-vector community}, PLoS Negl Trop Dis, 10 (2016),
  p.~e0004428.

\bibitem{jones2008global}
{\sc K.~E. Jones, N.~G. Patel, M.~A. Levy, A.~Storeygard, D.~Balk, J.~L.
  Gittleman, and P.~Daszak}, {\em Global trends in emerging infectious
  diseases}, Nature, 451 (2008), p.~990.

\bibitem{kilpatrick2012drivers}
{\sc A.~M. Kilpatrick and S.~E. Randolph}, {\em Drivers, dynamics, and control
  of emerging vector-borne zoonotic diseases}, The Lancet, 380 (2012),
  pp.~1946--1955.

\bibitem{Limul95}
{\sc M.~Li and J.~S. Muldowney}, {\em On r.a. smith's automonmous convergence
  theorem}, Rocky Mountain J. Math., 25 (1995), pp.~365--379.

\bibitem{MR95k:92022}
{\sc M.~Y. Li and J.~S. Muldowney}, {\em Global stability for the {SEIR} model
  in epidemiology}, Math. Biosci., 125 (1995), pp.~155--164.

\bibitem{1008.92032}
{\sc C.~McCluskey}, {\em {A model of HIV/AIDS with staged progression and
  amelioration.}}, Math. Biosci., 181 (2003), pp.~1--16.

\bibitem{1056.92052}
{\sc C.~McCluskey and P.~van~den Driessche}, {\em {Global analysis of two
  tuberculosis models.}}, J. Dyn. Differ. Equations, 16 (2004), pp.~139--166.

\bibitem{mitchell1995geographic}
{\sc C.~J. Mitchell}, {\em Geographic spread of aedes albopictus and potential
  for involvement in arbovirus cycles in the mediterranean basin}, Journal of
  Vector ecology, 20 (1985), pp.~44--58.

\bibitem{moon1970counting}
{\sc J.~Moon}, {\em Counting labeled trees,}.
\newblock Canadian Math, Monographs, 1970.

\bibitem{palmer2018dynamics}
{\sc C.~Palmer, E.~Landguth, E.~Stone, and T.~Johnson}, {\em The dynamics of
  vector-borne relapsing diseases}, Math. Biosci., 297 (2018).

\bibitem{paty2014large}
{\sc M.~Paty, C.~Six, F.~Charlet, G.~Heuz{\'e}, A.~Cochet, A.~Wiegandt,
  J.~Chappert, D.~Dejour-Salamanca, A.~Guinard, P.~Soler, V.~Servas,
  M.~Vivier-Darrigol, M.~Ledrans, M.~Debruyne, O.~Schaal, C.~Jeannin,
  B.~Helynck, I.~Leparc-Goffart, and B.~Coignard}, {\em Large number of
  imported chikungunya cases in mainland france, 2014: a challenge for
  surveillance and response}, Eurosurveillance, 19 (2014), p.~20856.

\bibitem{poletti2011transmission}
{\sc P.~Poletti, G.~Messeri, M.~Ajelli, R.~Vallorani, C.~Rizzo, and S.~Merler},
  {\em Transmission potential of chikungunya virus and control measures: the
  case of italy}, PLoS One, 6 (2011), p.~e18860.

\bibitem{rezza2007infection}
{\sc G.~Rezza, L.~Nicoletti, R.~Angelini, R.~Romi, A.~Finarelli, M.~Panning,
  P.~Cordioli, C.~Fortuna, S.~Boros, F.~Magurano, G.~Silvi, P.~Angelini,
  M.~Dottori, M.~Ciufolini, G.~Majori, and A.~Cassone}, {\em Infection with
  chikungunya virus in italy: an outbreak in a temperate region}, The Lancet,
  370 (2007), pp.~1840--1846.

\bibitem{rogers2006climate}
{\sc D.~Rogers and S.~Randolph}, {\em Climate change and vector-borne
  diseases}, Advances in parasitology, 62 (2006), pp.~345--381.

\bibitem{shroyer1986aedes}
{\sc D.~A. Shroyer}, {\em \textit{AEDES ALBOPICTUS} and arboviruses: A concise
  review of the literaturei}, Journal of the American Mosquito Control
  Association,  (1986).

\bibitem{tabata2016zika}
{\sc T.~Tabata, M.~Petitt, H.~Puerta-Guardo, D.~Michlmayr, C.~Wang,
  J.~Fang-Hoover, E.~Harris, and L.~Pereira}, {\em Zika virus targets different
  primary human placental cells, suggesting two routes for vertical
  transmission}, Cell host \& microbe, 20 (2016), pp.~155--166.

\bibitem{MR1993355}
{\sc H.~R. Thieme}, {\em Mathematics in population biology}, Princeton Series
  in Theoretical and Computational Biology, Princeton University Press,
  Princeton, NJ, 2003.

\bibitem{tumwiine2010host}
{\sc J.~Tumwiine, J.~Mugisha, and L.~Luboobi}, {\em A host-vector model for
  malaria with infective immigrants}, Journal of Mathematical Analysis and
  Applications, 361 (2010), pp.~139--149.

\bibitem{vega2013high}
{\sc A.~Vega-Rua, K.~Zouache, V.~Caro, L.~Diancourt, P.~Delaunay, M.~Grandadam,
  and A.-B. Failloux}, {\em High efficiency of temperate aedes albopictus to
  transmit chikungunya and dengue viruses in the southeast of france}, PLoS
  One, 8 (2013), p.~e59716.

\bibitem{0478.93044}
{\sc M.~Vidyasagar}, {\em {Decomposition techniques for large-scale systems
  with nonadditive interactions: Stability and stabilizability.}}, IEEE Trans.
  Autom. Control, 25 (1980), pp.~773--779.

\bibitem{wilder20142012}
{\sc A.~Wilder-Smith, M.~Quam, O.~Sessions, J.~Rocklov, J.~Liu-Helmersson,
  L.~Franco, and K.~Khan}, {\em The 2012 dengue outbreak in madeira: exploring
  the origins}, Euro Surveill., 19 (2014), p.~20718.

\bibitem{world2004world}
{\sc {World Health Organization} et~al.}, {\em The world health report: 2004:
  changing history}, 2004.

\bibitem{world2014global}
\leavevmode\vrule height 2pt depth -1.6pt width 23pt, {\em A global brief on
  vector-borne diseases}, 2014.

\end{thebibliography}

\end{document}